\documentclass[journal]{IEEEtran}
\usepackage{multicol} 

\usepackage{times,amsmath,amssymb,bm,bbm,tabularx,cases,dsfont}
\usepackage{hyperref}
\usepackage[capitalize]{cleveref}
\usepackage{comment}
\usepackage{cite}
\usepackage{amsthm}
\usepackage{textcomp}
\usepackage{stfloats}
\usepackage{mathtools}
\usepackage{algorithm,algpseudocode,setspace,etoolbox}
\usepackage{subfigure}
\usepackage{makecell}
\usepackage{blindtext}

\usepackage{graphicx}
\usepackage{xcolor}
\pdfoutput=1
\usepackage{bmpsize}

\usepackage[T1]{fontenc}

\crefformat{equation}{(#2#1#3)}
\crefrangeformat{equation}{(#3#1#4) --~(#5#2#6)}
\crefmultiformat{equation}{(#2#1#3)}%
{ and~(#2#1#3)}{, (#2#1#3)}{ and~(#2#1#3)}

\usepackage{textcomp}

\crefname{figure}{Fig.}{Fig.}
\crefname{table}{Table}{Table}

\theoremstyle{definition}

\newtheorem{lemma}{Lemma}
\newtheorem{proposition}{Proposition}

\algnewcommand{\Initialize}[1]{%
  \State \textbf{Initialize:}
  \State \hspace*{\algorithmicindent}\parbox[t]{0.8\linewidth}{\raggedright #1}
}
\newcommand{\EE}{{\mathbb{E}}}

\newcommand{\RR}{{\mathbb{R}}}

\newcommand{\LL}{{\mathbb{L}}}

\renewcommand{\mathbf}{\bm} 

\begin{document}
\bstctlcite{IEEEexample:BSTcontrol}

\title{Performance Analysis of Indoor VLC Network with Secure Downlink NOMA for Body Blockage Model}

\author{Tianji~Shen,~\IEEEmembership{Member,~IEEE}, 
Vamoua~Yachongka,~\IEEEmembership{Member,~IEEE},
Yuto~Hama,~\IEEEmembership{Member,~IEEE},
and~Hideki~Ochiai,~\IEEEmembership{Fellow,~IEEE}
\thanks{
  Tianji~Shen is with the B2B Division, 
  Hithink Royalflush Information Network Co. Ltd., Hangzhou, Zhejiang 310023, China.
  (E-mail:~\mbox{shentianji@myhexin.com})

  Vamoua~Yachongka is with
  the Department of Computer Science and Engineering, University of Texas at Arlington, Arlington, TX 76010 USA.
  (E-mail:~\mbox{va.yachongka@ieee.org})

  Yuto~Hama is with Department of Electrical and Computer Engineering, Yokohama National University, Yokohama, Kanagawa 240-8501, Japan.
  (E-mail:~\mbox{yuto.hama@ieee.org})

  Hideki~Ochiai is with 
  the Graduate School of Engineering, Osaka University, Osaka 565-0871, Japan 
  (E-mail:~\mbox{ochiai@comm.eng.osaka-u.ac.jp}).
  }
}
\makeatletter
\def\ps@IEEEtitlepagestyle{%
  \def\@oddfoot{\mycopyrightnotice}%
  \def\@oddhead{\hbox{}\@IEEEheaderstyle\leftmark\hfil\thepage}\relax
  \def\@evenhead{\@IEEEheaderstyle\thepage\hfil\leftmark\hbox{}}\relax
  \def\@evenfoot{}%
}
\def\mycopyrightnotice{%
  \begin{minipage}{\textwidth}
  \centering \scriptsize
  This work has been submitted to the IEEE for possible publication. Copyright may be transferred without notice, after which this version may no longer be accessible.
  \end{minipage}
}
\makeatother

\maketitle
\begin{abstract}
    In this work,
    we investigate
    the performance
    of indoor visible light communication~(VLC) networks 
    based on power domain non-orthogonal multiple access~(NOMA) 
    for mobile devices,
    where multiple legitimate users 
    are equipped with 
    photodiodes~(PDs).
    We propose a body blockage model 
    for both the legitimate users 
    and eavesdropper to address scenarios 
    where the communication links from transmitting light-emitting diodes (LEDs) 
    to receiving devices are blocked by the bodies of all parties.
    Furthermore, we propose a novel LED arrangement that improves secrecy without requiring knowledge of the channel state information (CSI) of the eavesdropper. 
    This arrangement reduces the overlapping areas covered by different LED units supporting distinct users. 
    We also suggest two LED transmission strategies, 
    i.e., 
    \emph{simple} 
    and 
    \emph{smart LED linking},
    and compare their performance with
    the conventional \emph{broadcasting}
    in terms of transmission sum rate and secrecy sum rate.
    Through computer simulations,
    we demonstrate the superiority of our proposed strategies
    to the conventional approach.
\end{abstract}

\section{Introduction}
\label{sec:intro}

\IEEEPARstart{V}{isible} light communication (VLC) 
is a promising technology 
for the beyond fifth-generation (B5G) and sixth-generation (6G) networks~\cite{chi_2020_vlc}. 
VLC operates in the visible light spectrum, 
utilizing the emitted light for both illumination and data transmission.
With its high potential data rate,
VLC has garnered significant interest~\cite{ndjiongue_vlc_2015}. Implemented through intensity modulation (IM) of light-emitting diodes (LEDs), VLC leverages existing infrastructure for illumination. It serves as an intriguing complement to radio frequency (RF) communication, particularly due to the need for strict RF radiation monitoring.

Non-orthogonal multiple access (NOMA) has gained attention for its high spectral efficiency compared to the conventional orthogonal multiple access (OMA) techniques~\cite{saito_2013_nonorthogonal,yang_2017_ontheoptimality,shen_2021_auav-enabled}. 
Unlike OMA, NOMA allocates the same time-frequency resources to multiple users,
maximizing the limited frequency spectrum utilization~\cite{akbar_2021_noma}.
This work focuses on power domain NOMA, 
employing multi-user superposition coding at the transmitter and successive interference cancellation (SIC) at the receivers~\cite{dai_noma_2018}. 
Previous studies, 
such as~\cite{choi_2016_allocation,cui_2016_outage}, 
have investigated power allocation algorithms for NOMA.

In addition,
physical layer security (PLS) facilitates 
secure communication 
by leveraging wireless channel randomness,
including fading, noise, and interference~\cite{wyner_1975_thewiretap}. 
PLS aims to utilize the stochastic properties of wireless channels 
to provide secure communication. 
It ensures secrecy when the channel of the legitimate user 
surpasses that of the eavesdropper,
making eavesdropping-resistant information transmission 
feasible without the need for symmetric keys 
or complex encryption/decryption algorithms~\cite{bloch_2008_wireless,liu_2017_physical,hamamreh_2019_classifications,mukherjee_2014_prnciples}.
Also,
PLS with VLC-enabled NOMA has been investigated
in~\cite{zhao_2018_noma,peng_2021_pls,shi_2022_secure,li_2022_secrecy}.
Specifically,
in~\cite{zhao_2018_noma},
a precise expression of secure outage probability was derived for 
a NOMA-enabled multi-user multiple-input single-output~(MISO) 
VLC system.
In~\cite{peng_2021_pls}, 
the secure transmission of MISO with NOMA under VLC was investigated, 
taking into account both known and unknown channel state information~(CSI) of eavesdroppers.
In~\cite{shi_2022_secure},
secrecy performance of a NOMA-enabled VLC
system in the presence of an active eavesdropper
was demonstrated.
It is shown in~\cite{li_2022_secrecy} that 
the pseudo user, 
who interrupts the eavesdropper's SIC process,
enhances secrecy performance in the NOMA-enabled VLC system.

\subsection{Literature Review and Motivation}
\label{subsec:related_work}

In recent years,
a variety of signal processing approaches
has been investigated for VLC-aided PLS systems
such as beamforming~\cite{mostafa_physical-layer_2015,liu_2020_beamforming,cho_2021_cooperativebeamforming},
precoding~\cite{arfaoui_secrecy_2018,pham_energy_2021},
artificial optical jamming~\cite{wang_optical_2018},
and polar code-based
secure coding~\cite{che_physical-layer_2018}.
Stochastic geometry theory
serves as a mathematical tool 
to determine the average secrecy capacity 
of the system,
considering a random distribution of users 
throughout the room~\cite{pan_secure_2017}.
Convex optimization theory enhances the secrecy performance of VLC systems subject to various illumination constraints, 
such as the LED optical power constraint~\cite{cho_physical_2019},
the peak power constraint 
to reduce eye damage~\cite{mostafa_physical-layer_2015},
and the light energy harvesting 
and dimming control constraint~\cite{liu_2020_beamforming}.
However,
most of the studies on PLS for VLC networks 
have focused on defending against external malicious eavesdroppers~\cite{pham_energy_2021,pan_secure_2017,cho_physical_2019},
some of which addressing the transmission of confidential information where users are unaware of information
not intended for them~\cite{arfaoui_secrecy_2018}.
These studies, 
based on the properties of VLC channels,
explored the use of beamforming 
and jamming techniques to enhance secrecy 
in the physical layer,
despite the high computational complexity involved. 
In practice,
however, 
the position of the passive eavesdropper is unknown
since they do not respond to transmission mechanisms.
Furthermore,
broadcasting systems
still face significant challenges
in implementing beamforming schemes
for VLC transmission.
With respect to LED arrangements, 
most works adopt
a square lattice alignment
to improve illuminance~\cite{zhang_2016_grouping}.
However, 
in broadcasting VLC,
such an arrangement may degrade communication reliability
since interference may occur in the overlapping areas served by multiple LED devices with each transmitting their own signal.


Furthermore, 
previous researches on indoor VLC systems ignore
the effects of body blockage from users,
with few exceptions such as~\cite{singh_2021_vlc,aboagye_2023_ris,beysens_2020_exploiting}.
In~\cite{singh_2021_vlc},
the authors investigated the
cylinder model of user bodies
without considering the angles.
Reference \cite{aboagye_2023_ris}
investigates the use of reconfigurable intelligent surfaces
(RIS) in order to avoid body blockage in VLC system.
The authors of~\cite{beysens_2020_exploiting} 
investigated body blockage modeling 
by dividing the orientation angle 
into several parts and assessing each part for blockage.


\subsection{Contributions}
\label{subsec:novelty}
Inspired by the above-mentioned state-of-the-art studies,
our work focuses on VLC-enabled NOMA-based PLS systems
without any knowledge of the eavesdropper.
The main contributions\footnote{
The preliminary version of this work was presented 
at the 2023 IEEE Wireless Communications and Networking Conference (WCNC) 
in \cite{shen_vlc_2022}.
In the present manuscript, 
we additionally propose a new strategy 
(i.e., {\it smart~LED~linking strategy})
as well as a new power allocation scheme based on an estimated maximum sum rate in order to improve the secrecy performance.} 
of this work are summarized as follows:
\begin{itemize}
  \item
  The aforementioned works (e.g.,~\cite{beysens_2020_exploiting,soltani_2019_modeling}) 
  only partially address the modeling of body blockage.
  However,
  the modeling methods for \emph{body blockage} 
  in these works are not universally applicable,
  as they rely on statistical methods 
  to approximate the channel state from LED sources to PD receivers.
  Therefore, 
  we propose a geometrical-based general blockage model in the presence of a passive eavesdropper.
  It accounts for scenarios 
  where user bodies may block the light path 
  from the LED to the PD,
  determining blockage events through geometric calculations.

  \item 
  We propose a new power allocation scheme based on 
  NOMA to maximize the transmission sum rate of legitimate users, 
  taking into account only their locations 
  and statistical radiation angles. 
  We assume that the system is unaware 
  of the azimuth and polar angles of the devices,
  crucial for determining 
  if the light obstruction is present. 
  Furthermore,
  two power allocation schemes are compared through computer simulations.

  \item 
  We propose a new LED arrangement that places LEDs 
  in the form of 
  an equilateral triangle lattice 
  to reduce the overlapping area of different LED units. 
  Furthermore, we introduce
  two transmission strategies
  for this LED arrangement,
  called
  \emph{simple}
  and
  \emph{smart~LED~linking strategies},
  so as to investigate the trade-off between 
  the transmission and secrecy sum rates.
\end{itemize}

\subsection{Paper Organization}
The rest of this paper is organized as follows.
\cref{sec:system-model} describes the system model including the channel model as well as the body blockage model considered throughout this work,
and \cref{sec:signal} describes the NOMA signal models and performance metrics.
In~\cref{sec:strategies},
the proposed transmission strategies of secure NOMA in VLC
are described.
Simulation results are presented and discussed in~\cref{sec:simu}.
Finally, \cref{sec:conclusion} concludes this work.

\subsection*{Notations}

\begin{table}[!t]
  \centering
  \small
  \caption{Notations}
  \label{tab:notation}
  \begin{tabularx}{\columnwidth}{|c|X|}
    \hline
    {\bf Notation} & {\bf Description} \\ \hline \hline
    $[x]^+$ & An operator is defined as $[x]^+ = \max(x,0)$. \\ \hline
    $\EE[\cdot]$ & Expectation operator. \\ \hline
    ${\cal N}(\mu, \sigma^2)$ & 
    The circular-symmetric Gaussian random 
    distribution
    with mean $\mu$ and variance $\sigma^2$. \\ \hline
    ${Q}(\cdot)$ &
    The Gaussian $Q$ function,
    defined as 
    $Q(x) \triangleq \int_x^{\infty}\frac{e^{-t^2/2}}{\sqrt{2\pi}}dt$
    \cite[(B.20)]{goldsmith}. \\ \hline
    $\mathbbm{1}
    \left[{\cal X}\right]$ & 
    A binary indicator function for an event ${\cal X}$ that is equal to $1$ if the event ${\cal X}$ occurs and $0$ otherwise. \\ \hline
    $\|{\bm x}\|$ &
    Euclidean norm of the vector ${\bm x} = (x_1, x_2, \ldots, x_M)$,
    i.e.,
    $\| {\bm x}\| = \sqrt{\sum_{m=1}^M x_m^2}$\\ \hline
  \end{tabularx}
\end{table}
The notations used in this paper are listed in \cref{tab:notation}.

\section{System and Channel Model}
\label{sec:system-model}

In this section, we describe our VLC system model
as well as
VLC channel model.

\subsection{VLC System Model}

\begin{figure}[!ht]
  \centering
  \includegraphics[clip,width=\columnwidth]{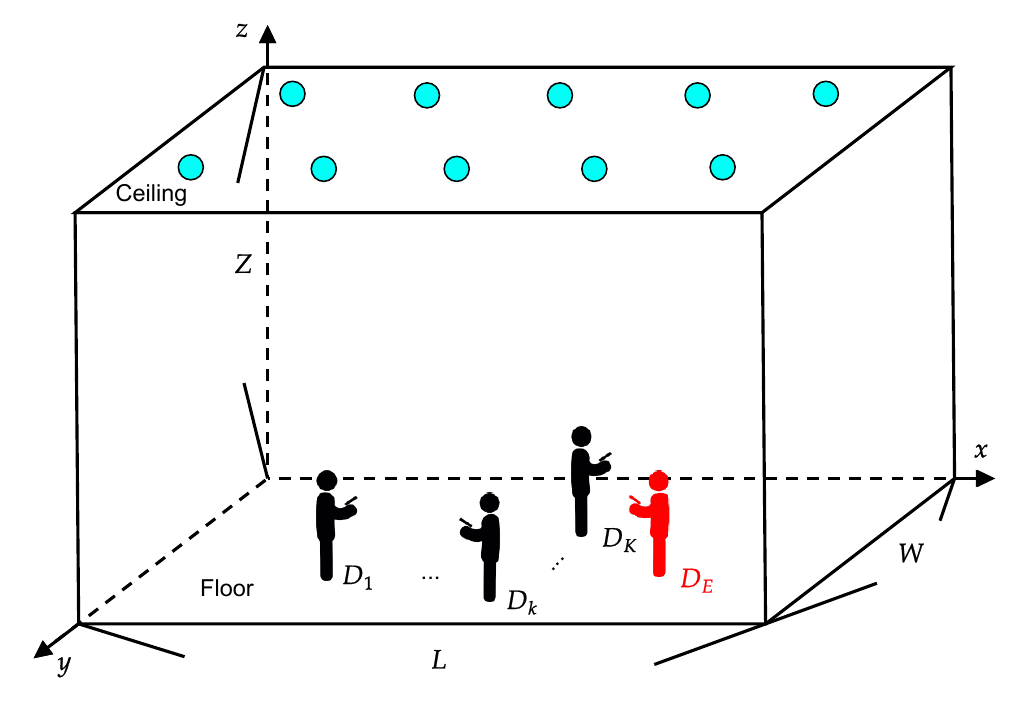}
  \caption{An illustration of the system model. 
  Cyan dots represent LED transmitters,
  $D_1, \cdots, D_K$ denote the positions of PDs of the legitimate users, 
  and 
  $D_E$ is 
  the position 
  of the eavesdropper.}
  \label{fig:system_model}
\end{figure}

As shown in~\cref{fig:system_model},
we consider an indoor downlink scenario
in the room size of length~$L$, width~$W$, and height~$Z$,
where $N$~LED sources
serve $K$~legitimate users
with the index set~${\cal K} = \{1,2,\ldots,K\}$, 
as well as an eavesdropper denoted by $D_E$.
We consider that each of the legitimate users and the eavesdropper has a single PD receiver.
We assume that the system is not aware of the location of the eavesdropper.
Moreover,
we assume that all the bodies of users as well as the eavesdropper are modeled as cylinders
with radius $r$ as 
in~\cite{komine_2005_shadowing,chen_2017_performance,beysens_2020_exploiting}.
Also, let $S_n = (x_{S_n},y_{S_n},Z)$
denote the position corresponding to the $n$th LED source 
in the Cartesian coordinates
for $n \in {\cal I} \triangleq \{1,2,\ldots,N\}$.
We define $D_{k} = (x_{D_{k}},y_{D_{k}},z_{D_{k}})$
with $k \in {\cal K}$
as the positions representing the corresponding PD of legitimate users,
and $D_E = (x_{E},y_{E},z_{E})$
representing the corresponding PD of the eavesdropper.
Furthermore,
we define
${\cal K}' \triangleq {\cal K} \cup E$ as the set of indices corresponding to all the PD receivers in the room.

In addition, 
we assume that the height of the device $z_{D_k}$
and eavesdropper $z_E$
are defined as the constant ratio $\nu \in (0,1)$
of the height of the users,
i.e.,
$z_{D_k} = \nu {\sf h}_{D_k}$
and $z_{D_E} = \nu {\sf h}_{E}$.

\subsection{Channel Model}

In this subsection, we define the VLC channel between an LED-PD transceiver pair
with emphasis on the effect of the users' bodies.

\subsubsection{VLC Channel}

\begin{figure}[!ht]
  \centering
  \includegraphics[clip,width=\columnwidth]{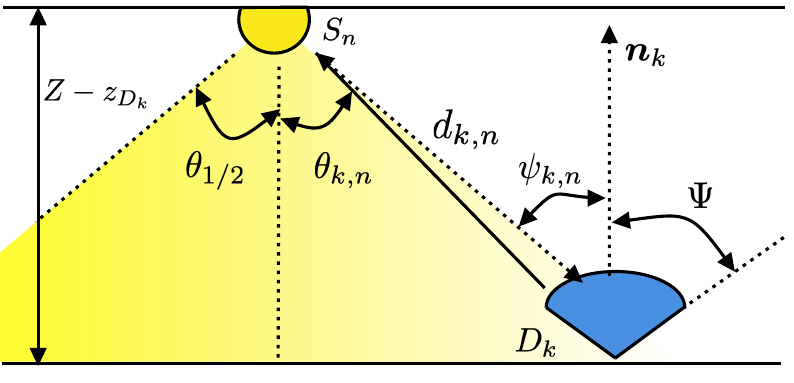}
  \caption{An illustration of the
    VLC path 
    between
    the $n$th LED and the $k$th PD,
    where
    $\theta_{k,n}$ is the radiation angle between $S_n$ and $D_k$,
    $\theta_{1/2}$ is the half illuminance angle of the LED transmitter,
    $\psi_{k,n}$ is the incidence angle of link between $S_n$ and $D_k$,
    $\Psi$ is the received field of view (FoV) of the PD receiver,
    and
    ${\bm n}_k$ is the normal direction of the photosensitive surface of the $k$th PD.}
  \label{fig:vlc_pathgain}
\end{figure}

\begin{figure}[!t]
  \centering
  \includegraphics[clip,width=\columnwidth]{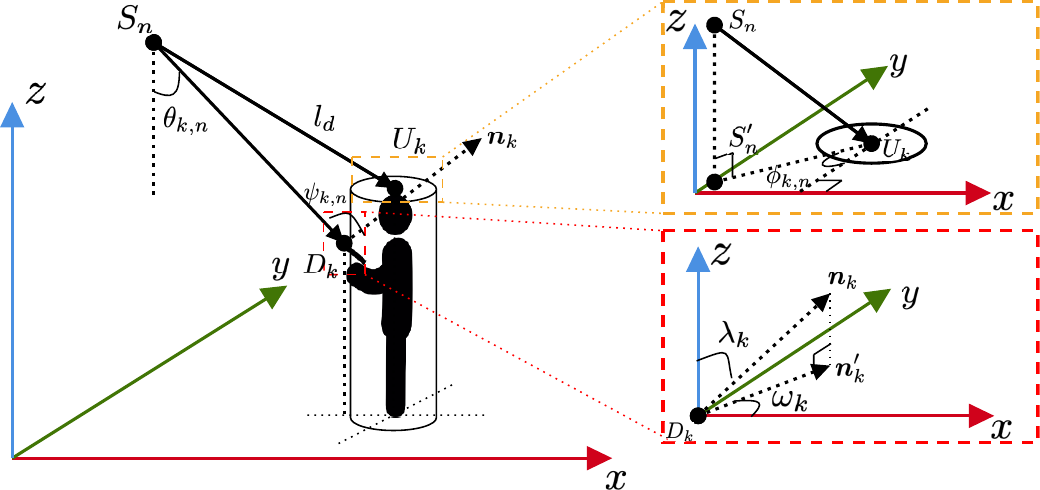}
  \caption{The geometry of visible light communications with $S_n$, $D_k$ and $U_k$.
  $S_n'$ is the projection point of $S_n$ onto the plane that 
  contains the top base of cylinder,
  ${\bm n}_k$ is the normal direction of the photosensitive surface of $D_k$ with its projection at XOY plane is given by ${\bm n}_k$,
  $\omega_k$ is the azimuth angle between ${\bm n}_k$ and $x$-axis,
  $\lambda_k$ is the polar angle of $D_k$,
  $\theta_{k,n}$ is the radiation angle from $S_n$ to $D_k$,
  $\psi_{k,n}$ is the incidence angle of the light beam from $S_n$ to $D_k$,
  and $\phi_{k,n}$ is the azimuth angle between the vector $\protect\overrightarrow{S_n U_k}$ and $y$-axis.}
  \label{fig:geometric_para}
\end{figure}

In~\cref{fig:vlc_pathgain},
we illustrate the geometric relationship between LED $S_n$ and PD $D_k$ 
with the normal direction of its photosensitive surface ${\bm n}_k$.
Most of the existing studies on VLC systems consider the line-of-sight~(LoS) component in 
VLC channel models and ignore the non-line-of-sight~(NLoS) component
due to a huge diffuse component gap between LoS and NLoS components~\cite{kahn_1997_wireless,
  komine_2004_vlc,
  yin_2016_performance}.
Subsequently,
assuming that the considered LEDs have Lambertian emission patterns,
the channel gain
between 
$S_n$
and 
$D_k$
with $k \in {\cal K}'$
is expressed as~\cite{kahn_1997_wireless}
\begin{align}
  \label{eq:channel_gain_user}
  &h_{k,n}
  \left(\theta_{k,n},\psi_{k,n}\middle|\Psi\right)
  \nonumber\\
  &= 
  A
  \frac{(m+1) R_{\mathrm{PD}}}{2\pi} 
  \cos^m(\theta_{k,n})
  \frac{\cos(\psi_{k,n})}{d_{k,n}^2}
  g\left(\psi_{k,n}\middle| \Psi\right)
  \nonumber\\
  & \quad \cdot
  \prod_{l \in {\cal K}'}
  \mathbbm{1}
  \left[{\cal E}(n,k,l)\right],
\end{align}
where
the notations introduced in~\cref{eq:channel_gain_user}
are listed in \cref{tab:2}.

\begin{table}
  \small
  \caption{Parameter Definitions in Equation \cref{eq:channel_gain_user}}
  \label{tab:2}
  \begin{tabularx}{\columnwidth}{|c|X|}
    \hline
    {\bf Notation} & {\bf Description} \\ \hline \hline
    $A$ & The effective area of the photosensitive surface of PD receivers. \\ \hline
    $\theta_{k,n}$ & The radiation angle from $S_n$ to $D_{k}$. \\ \hline
    $\psi_{k,n}$ & The incidence angle from $S_n$ to $D_{k}$. \\ \hline
    $\Psi$ & The received field of view (FoV) of PD receivers. \\ \hline
    $m$ & The order of the Lambertian emission
    with half illuminance angle of LED transmitter $\theta_{1/2}$,
    defined as $m = \frac{-\ln(2)}{\ln\left(\cos\left(\theta_{1/2}\right)\right)}$. \\ \hline
    $d_{n,k}$ & The Euclidean distance between $S_n$ and $D_{k}$. \\ \hline
    $g\left(\psi_{n,k} \middle| \Psi\right)$ & The gain of the optical concentrator with 
    reflective index $\eta \in (1,2)$,
    expressed as
      \[g\left(\psi_{k,n}\middle| \Psi\right)
      =
      \begin{cases}
        \frac{\eta^2}{\sin^2(\Psi)}, &~\mbox{if}~~0 \leq \psi_{k,n} \leq \Psi,\\
        0, &~\mbox{if}~~\psi_{k,n} > \Psi.
      \end{cases}\]
    \\ \hline
  \end{tabularx}
\end{table}
In \cref{eq:channel_gain_user}, 
${\cal E}(n, k, l)$
represents the event
that the link from 
the $n$th source $S_n$ to 
the $k$th user $D_k$ 
is \emph{not} blocked by the body of the $l$th user $D_{l}$.
In order to facilitate the understanding of these geometric parameters,
we depict the geometric relations of $S_n$, $D_k$, 
and $U_k$ in~\cref{fig:geometric_para},
where $U_k = (x_{U_k},y_{U_k},{\sf h}_{D_k})$ is the position of the top center of the cylinder
for $k \in {\cal K}'$.
The vector 
from $D_k$ to $U_k$,
$\overrightarrow{D_k U_k} \triangleq
  \left(
    x_{U_k} - x_{D_k},
    y_{U_k} - y_{D_k},
    {\sf h}_{D_k} - z_{D_k}
  \right)$,
can be expressed
as
\begin{align}
  \label{eq:relation_dk_uk}
  \overrightarrow{D_k U_k}
  &=
  \overrightarrow{S_n U_k} - \overrightarrow{S_n D_k}
  \nonumber \\
  &=
  \left(
    \sqrt{l_d^2- z_{D_k}^2} \cos \omega_k,
    \sqrt{l_d^2- z_{D_k}^2} \sin \omega_k,
    \Delta z_{D_k}
  \right),
\end{align}
where $l_d$
is
the distance between the center of cylinder and the PD of its mobile device,
$\Delta z_{D_k} \triangleq {\sf h}_{D_k}-z_{D_k}$,
and
$\omega_{k} \in [-\pi,\pi)$
is the azimuth angle between the line from $S_n$ to $D_k$
and $y$-axis,
i.e., the angle between the XOY projection of the direction vector ${\bm n}_{k}$ and $x$-axis~\cite{chi_2018_vlc}.
Moreover,
the incidence angle $\psi_{k,n}$ in~\cref{eq:channel_gain_user}
can be obtained by the positions of $S_n$ and $D_{k}$,
azimuth angle $\omega_k$,
and polar angle $\lambda_k \in [0, \frac{\pi}{2}]$,
i.e.,~\cite{soltani_2019_modeling}
\begin{align}
  \label{eq:omega}
  &\cos \psi_{k,n} 
  \nonumber\\
  &= 
  \left(
    \frac{
      x_{S_n} - x_{D_k}
    }
    {
      \left\|\overrightarrow{D_k S_n}\right\|
    }
  \right)
  \sin \lambda_{k}
  \cos \omega_{k}
  \nonumber\\
  &\quad +
  \left(
    \frac{
      y_{S_n} - y_{D_k}
    }
    {
      \left\|\overrightarrow{D_k S_n}\right\|
    }
  \right)
  \sin \lambda_{k}
  \sin \omega_{k}
  +
  \left(
    \frac{
      Z - z_{D_k}
    }
    {
      \left\|\overrightarrow{D_k S_n}\right\|
    }
  \right)
  \cos \lambda_k,
\end{align}
where
$\overrightarrow{D_k S_n}$
is the vector from $D_k$ to $S_n$,
i.e.,
$\overrightarrow{O S_n} - \overrightarrow{O D_k}$
with the origin $O$.
Additionally,
based on the statistical model of the orientation of mobile devices 
with \emph{walking-activity} situation
in~\cite{soltani_2019_modeling},
we set $\omega_k$ as a uniform random variable
over the interval from $-\pi$ to $\pi$,
and $\lambda_k$
as the truncated Laplace distribution 
with a mean of $\mu = 0.518$
and a standard deviation of $\sigma = 0.136$.\footnote{
  The range of $\lambda_k$
  is strictly limited to $[0, \pi/2]$.
  Therefore,
  the probability distribution function (PDF)
  can be given by 
  $f_{\lambda_k}(\lambda_k)
  =
  \frac{
    \exp
    \left(
      -\frac{
        \lambda_k - \mu
      }
      {
        b
      }
    \right)
  }
  {
    2b
    \left(
      G_{
        \frac{\pi}{2}
      }
      -
      G_0
    \right)
  }$,
  where $b = \sqrt{\sigma^2/2}$,
  $G_0 = \frac{1}{2}\exp\left(-\frac{\mu}{b_{k}}\right)$,
  and
  $G_{\frac{\pi}{2}} = 1 - \frac{1}{2}\exp\left(-\frac{\frac{\pi}{2}-\mu}{b}\right)$.
  As a given pair of $\mu$ and $\sigma$,
  we can derive that $G_{\frac{\pi}{2}} \approx 1$
  and $G_0 \approx 0$,
  and it can be shown that $\int_{0}^{\frac{\pi}{2}}
  f_{\lambda_k}(\lambda) d\lambda \approx 1$
  \cite[eqs. (10) and (11)]{soltani_2019_modeling}.
}

\subsubsection{Body Blockage Model}
\label{sec:body_model}

\begin{figure}[t]
  \centering
  \subfigure[Three-dimensional (3D) schematic diagram.]{
    \includegraphics[clip,width=.4\textwidth]{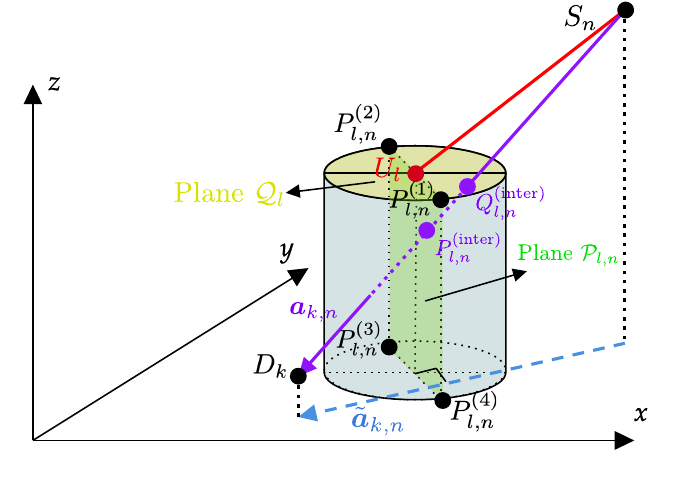}
    \label{fig:cylinder_shadow}
  }
  \subfigure[Two-dimensional (2D) plane (XOY) schematic diagram.]{
    \includegraphics[clip,width=.26\textwidth]{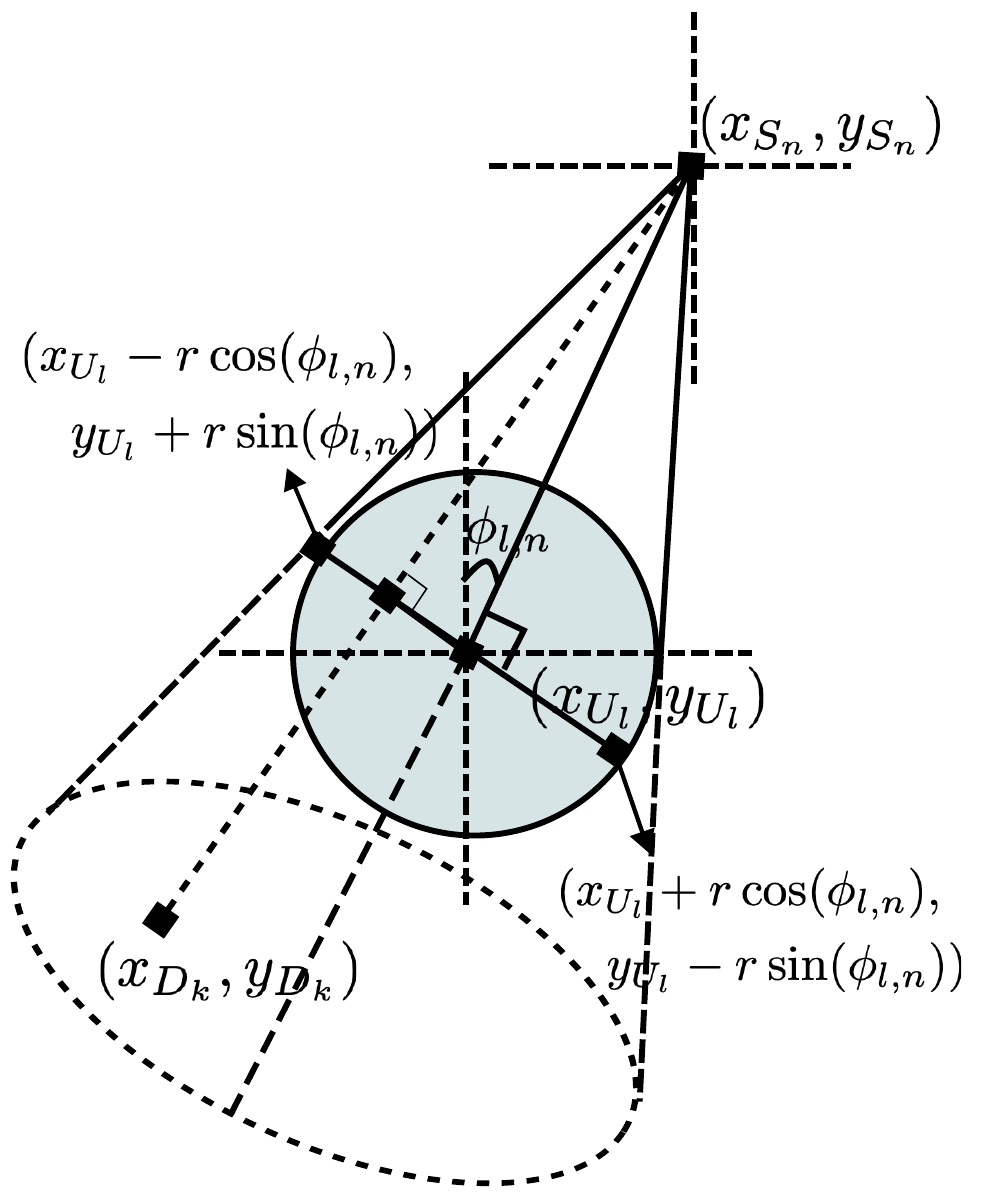}
    \label{fig:plane_shadow}
  }
  \caption{A mathematical explanation of the body blockage model in this work.
  The $n$th LED light is positioned at $S_n(x_{S_n}, y_{S_n}, Z)$, while the PD receiver for the $k$th user is located at $D_k(x_{D_k}, y_{D_k}, z_{D_k})$.
  A cylinder, who is modeled as the receiver holder of $D_{l}$ with $l \in {\cal K}'$, is centered at $U_{l}(x_{U_{l}}, y_{U_{l}}, {\sf h}_{D_l})$ on its top plane, with an azimuth angle $\phi_{l,n}$ between itself and the $n$th LED source.
  The light vector ${\bm a}_{k, n}$ projects onto the XOY plane as $\tilde{\bm a}_{k, n}$.
  The plane ${\cal P}_{l, n}$, which passes through the point $U_{l}$, is perpendicular to the vector $\tilde{\bm a}_{k, n}$.
  Additionally, the vector ${\bm a}_{k, n}$ intersects with the planes ${\cal Q}_{l}$ and ${\cal P}_{l, n}$ at the points $Q_{l, n}$ and $P_{l, n}$, respectively.
  }
  \label{fig:shadow}
\end{figure}
The presence of the body of users as well as the eavesdropper
inside an indoor environment may certainly 
affect the channel characteristics 
and the communication performance of a VLC system. 
In fact, the channel gain of the PD devices is affected not only by its geometric parameters
but also by the shadows of the human bodies.
Subsequently,
we focus on the event ${\cal E}(n,{k},{l})$ in~\cref{eq:channel_gain_user}.
The indicator function 
$\mathbbm{1} \left[{\cal E}(n,k,l)\right]$
plays a pivotal role in determining whether this event occurs, returning 1 if it does and 0 otherwise.
In what follows, we describe the role of this function in detail.
As illustrated in~\cref{fig:shadow},
in order to simplify notations,
let us define 
${\bm a}_{k,n} = \overrightarrow{S_n D_k} = (x_{D_k}-x_{S_n}, y_{D_k}-y_{S_n}, z_{D_k}-Z)$
as a vector of light
from $S_n$ to $D_k$
with 
the projection on the XOY plane 
$\tilde{\bm a}_{k,n} 
= 
(x_{D_k}-x_{S_n}, y_{D_k}-y_{S_n},0)$.
Moreover,
we define two planes ${\cal Q}_{l}$ and ${\cal P}_{l,n}$,
where ${\cal Q}_{l}$ contains 
the point $U_{l}$
with
its unit normal vector given by
${\bm q} = (0,0,1)$,
and ${\cal P}_{l,n}$ 
contains the point $U_{l}$ with its unit normal vector 
given by
${\bm p}_{l,n} = \alpha \tilde{\bm a}_{k,n}$
for
a real-valued factor $\alpha$.
The projection area of the vertical cylinder with the four vertices
on the plane ${\cal P}_{l,n}$ is a rectangle, 
with their coordinates given by
\begin{subequations}
  \label{eq:B}
  \begin{align}
    P_{l,n}^{(1)} 
    &= \left(x_{U_{l}} + r \cos \phi_{l,n}, 
    y_{U_{l}} - r \sin \phi_{l,n}, 
    {\sf h}_{D_l}
    \right),\\
    P_{l,n}^{(2)}
    &= \left(x_{U_{l}} - r \cos \phi_{l,n}, 
    y_{U_{l}} + r \sin \phi_{l,n}, 
    {\sf h}_{D_l}
    \right),\\
    P_{l,n}^{(3)}
    &= \left(x_{U_{l}} - r \cos \phi_{l,n}, 
    y_{U_{l}} + r \sin \phi_{l,n}, 
    0
    \right),\\
    P_{l,n}^{(4)}
    &= \left(x_{U_{l}} + r \cos \phi_{l,n}, 
    y_{U_{l}} - r \sin \phi_{l,n}, 
    0
    \right),
  \end{align}
\end{subequations}
where 
$\phi_{l,n}$ is
the azimuth angle between $S_n$
and $U_{l}$,
defined as
\begin{align}
  \label{eq:phi_{U_k}}
  \phi_{l,n} = \arctan 
  \frac{
    \left|
    y_{S_n} - y_{U_{l}}
    \right|
  }
  {\left|
    x_{S_n} - x_{U_{l}}
    \right|
  }.
\end{align}
Next,
the intersection point of a line and a plane satisfies the following Lemma:
\begin{lemma}
  \label{lemma:1}
  \begin{figure}[t]
    \centering
    \includegraphics[clip,width=.8\columnwidth]{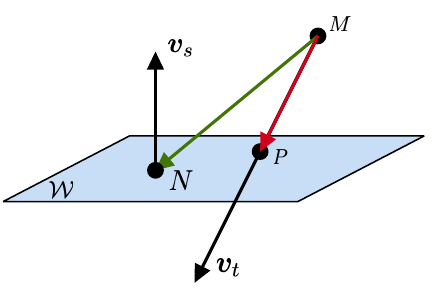}
    \caption{Illustration for the relation of intersection point of vectors and plane ${\cal W}$ in Lemma 1.
    Points $N$ and $P$ are on plane ${\cal W}$,
    ${\bm v}_s$ is the normal vector to plane ${\cal W}$,
    and ${\bm v}_t$ is the vector have same direction as $\protect\overrightarrow{MP}$ 
    (i.e., ${\bm v}_t = a \protect\overrightarrow{MP}$ with real factor $a$).}
    \label{fig:lemma1}
  \end{figure}
  Suppose that 
  the vector $\mathbf{v}_{t} = (x_{t},y_{t},z_{t})$
  in Euclidean space
  passes through
  the point $M(x_m,y_m,z_m)$
  and 
  the point $P(x_p,y_p,z_p)$
  on the
  plane ${\cal W}$
  as illustrated in \cref{fig:lemma1}.
  If
  the plane ${\cal W}$ also includes the point $N(x_n,y_n,z_n)$ 
  with one of its normal vectors
  $\mathbf{v}_s = (x_{s},y_{s},z_{s})$,
  the vector $\overrightarrow{MP}$ should satisfy
  \begin{align}
    \label{eq:lemma_1}
    \overrightarrow{MP} = 
    \frac{\overrightarrow{MN} \cdot \mathbf{v}_s}
    {\mathbf{v}_s \cdot \mathbf{v}_{t}}
    \mathbf{v}_{t},
  \end{align}
  where the inner product of ${\bm v}_s$ and ${\bm v}_{t}$ satisfies
  \begin{align}
    {\mathbf{v}_s \cdot \mathbf{v}_{t}} \neq 0.
  \end{align}
\end{lemma}
\begin{proof}
    See Appendix~\ref{sec:appendix}.
\end{proof}

We define $Q_{l,n}^{({\rm inter})}$ as the intersection of vector 
${\bm a}_{k,n}$ and plane ${\cal Q}_{l}$,
and $P_{l,n}^{({\rm inter})}$ as the intersection of vector 
${\bm a}_{k,n}$ and plane ${\cal P}_{l,n}$.
Their coordinates can be derived from \cref{lemma:1} as
\begin{align}
  \overrightarrow{S_n Q_{l,n}^{({\rm inter})}} &= 
  \frac{
    \overrightarrow{S_n U_{l}}
    \cdot
    {\bm a}_{k,n}
  }
  {{\bm a}_{k,n} \cdot {\bm q}}
  {\bm q}, \\
  \overrightarrow{S_n P_{l,n}^{({\rm inter})}} &= 
  \frac{
    \overrightarrow{S_n U_{l}}
    \cdot
    {\bm a}_{k,n}
  }
  {{\bm a}_{k,n} \cdot {\bm p}_{l,n}}
  {\bm p}_{l,n}.
\end{align}
Therefore,
the event ${\cal E}_k(n, {k}, {l})$
in
\cref{eq:channel_gain_user}
can be determined by
the following proposition:

\begin{proposition} 
  \label{pro-2}
 
  The path ${\bm a}_{k,n}$
  is blocked by the cylinder body of user $U_{l}$,
  (i.e., the event ${\cal E}(n, {k}, {l})$ in~\cref{eq:channel_gain_user} is \emph{false}),
  if either of the following two conditions is satisfied:
  \begin{enumerate}
    \item 
    The intersection point $Q_{l,n}^{({\rm inter})}$ should lie in 
    the top base of cylinder ${\cal Q}_{l}$,
    i.e.,
    it
    should satisfy that
    \begin{align}
      \|Q_{l,n}^{({\rm inter})}-U_{l}\| \leq r.
      \label{eq:cond_1}
    \end{align}

    \item
    The intersection point $P_{l,n}^{({\rm inter})}$
    should be in the area of 
    the rectangle formed by
    the points
    $P_{l,n}^{(1)}$, $P_{l,n}^{(2)}$, $P_{l,n}^{(3)}$,
    and $P_{l,n}^{(4)}$,
    i.e.,
    it should satisfy both of the following inequalities
\end{enumerate}
    \begin{align}
      \left(
        \overrightarrow{P_{l,n}^{(1)} P_{l,n}^{(2)}} 
        \times
        \overrightarrow{P_{l,n}^{(1)} P_{l,n}^{({\rm inter})}} 
      \right)
      \cdot
      \left(
        \overrightarrow{P_{l,n}^{(3)} P_{l,n}^{(4)}} 
        \times
        \overrightarrow{P_{l,n}^{(3)} P_{l,n}^{({\rm inter})}} 
      \right)
      &\geq 0, 
      \nonumber\\
      \left(
        \overrightarrow{P_{l,n}^{(2)} P_{l,n}^{(3)}} 
        \times
        \overrightarrow{P_{l,n}^{(2)} P_{l,n}^{({\rm inter})}} 
      \right)
      \cdot
      \left(
        \overrightarrow{P_{l,n}^{(4)} P_{l,n}^{(1)}} 
        \times
        \overrightarrow{P_{l,n}^{(4)} P_{l,n}^{({\rm inter})}} 
      \right)
      &\geq 0.
      \label{eq:cond_2}
    \end{align}
  Otherwise,
  it is not blocked by the cylinder (i.e., the event ${\cal E}(n, {k}, {l})$ 
  in \cref{eq:channel_gain_user}
  is \emph{true}).
\end{proposition}
\section{NOMA Signal Models and Performance Metrics}
\label{sec:signal}

In this section,
we describe the NOMA signal model for our proposed VLC transmission system
and introduce performance metrics of the transmission and secrecy sum rate.

\subsection{NOMA Signal Model}

We assume that the transmission signals by LED transmitters are 
the direct current (DC) biased 
signals~\cite{zhang_2016_grouping}.
For convenience, we define a new set $\mathcal{K}_n \subset \mathcal{K}$,
which represents the indices of users
that receive the signal transmitted from the $n$th LED $S_n$.
The transmitted symbol from the $n$th LED,
denoted by~$x_n \in \RR$,
is 
a Gaussian signal
with a fixed bias current $I_{\mathrm{DC}} \in \RR_{+}$
that is applied for the purpose of illumination.
Following the definition in~\cite[eq.~(2)]{zhang_2016_grouping}, 
we have\footnote{
    For the signal power of VLC transmission,
    $I_{\rm DC}$ is a large constant value.
    This assumption is commonly used 
    in the studies related to VLC,
    e.g.,~\cite{zhang_2016_grouping,kahn_1997_wireless,li_2023_adaptive}.
}
\begin{align}
  \label{eq:x_n}
  x_{n} = \sum_{k \in {\cal K}_n} \sqrt{P_s \beta_{k}({\cal K}_n)} 
  s_{n,k}
  + I_{\mathrm{DC}},
\end{align}
where $P_s$ represents the total transmission power at each LED
and $s_{n,k}$ is the modulated message signal intended for $D_k$ by $S_n$,
which is assumed to be a Gaussian random variable with zero-mean
and unit variance\footnote{
  Provided that $I_{\rm DC}$ is sufficiently large,
  the probability of negative $x_n$ 
  could be ignored.
}, 
i.e., $s_{n,k} \sim {\cal N}(0,1)$.
Furthermore,
$\beta_k({\cal K}_n)$ represents the allocated power ratio for~$D_{k}$
and fulfills that
\begin{align}
  \sum_{k \in {\cal K}_n} \beta_k({\cal K}_n) = 1.
\end{align}

Therefore,
after removing the DC bias,
the received signal by $D_k$ for
$k \in {\cal K}'$
can be
formulated
as
\begin{align}
  y_k
  &=
  \sum_{n=1}^N
  {h}_{k,n}\left(\theta_{k,n},\psi_{k,n}\middle|\Psi\right)
  \sum_{k_n \in {\cal K}_n}
  \sqrt{P_s \beta_{k_n}({\cal K}_n)} x_n
  + 
  n_{k},
\end{align}
where
$h_{k,n}\left(\theta_{k,n},\psi_{k,n} \middle| \Psi\right)$
is the channel coefficient from $S_n$ to $D_{k}$
defined in~\cref{eq:channel_gain_user}
and
$n_{k}$
is the real additive white Gaussian noise (AWGN)
with variance $\sigma_N^2$
at
$D_k$,
i.e.,
$n_k \sim {\cal N}(0, \sigma_N^2)$.

\subsection{Transmission and Secrecy Sum Rate} 

We assume that each device can decode its information with perfect SIC~\cite{kizilirmak_2015_perfectSIC},
i.e., 
there is no 
residual interference
after the SIC process.
The device first detects the symbol 
with the highest received power
and
then starts the SIC process by
treating the other symbols as 
additional Gaussian noise terms.
The transmission rate\footnote{
  The units of transmission rate as well as secrecy rate are defined as bit/Hz/sec
  \cite{feng_2019_joint}.
} 
of the $k$th PD $D_k$ for $k \in {\cal K}_n$,
defined as $R_{D_k}$,
is evaluated 
as~\cite{feng_2019_joint}
\begin{align}
  R_{D_k}\left(\theta_{k,n},\psi_{k,n}\right) &=
  \frac{1}{2} 
  \log_2
  \left(
    1 + \gamma_{k}
    \left(
      \theta_{k,n},\psi_{k,n},
      {\cal K}_n
    \right)
  \right),
\end{align}
where 
the scaling factor $\frac{1}{2}$ is due to the 
baseband signal transmission
and
$\gamma_k(\theta_{k,n},\psi_{k,n},{\cal K}_n)$ is
the received signal-to-interference-plus-noise ratio (SINR) at
the $k$th user (or eavesdropper)
under the assumption of the perfect SIC, defined in
\cref{eq:gamma_uk} located at the top of this page,
with
$|{\cal K}_n|$ representing the cardinality of the set ${\cal K}_n$.
We define a new set for all the indices of LED light signal received by the $k$th user (or eavesdropper) in the set ${\cal K}_n$
as 
$\mathcal{I}_k = \{n~|~k \in \mathcal{K}_n, n \in \mathcal{I}\}$,
e.g.,
if the $k$th user received the signals transmitted by LED sources with indices 1 and 2,
then ${\cal I}_k = \{1,2\}$.
Also,
$\mathcal{I}_k^{\tilde{n}}$
is
defined in the same manner as $\mathcal{I}_k$
with an additional condition $\tilde{n} \neq n$,
i.e.,
$\mathcal{I}_k^{\tilde{n}} = \{n~|~k \in \mathcal{K}_n, n \in \mathcal{I}, n \neq \tilde{n}\}$.

The eavesdropper wiretaps information sent from $S_n$ to $D_k$ 
during the downlink transmission period and 
employs the same decoding process as that performed at legitimate users.
The transmission rate, 
which can be wiretapped by the eavesdropper from the data intended for the $k$th user via the eavesdropper's channel, 
defined as $R_{E,D_k}$,
is evaluated as
\begin{align}
  R_{E,D_k}\left(\theta_{E,n},\psi_{E,n}\right) = 
  \frac{1}{2}
  \log_2
  \left(
    1
    +
    \gamma_{k}(
      \theta_{E,n},
      \psi_{E,n},
      {\cal K}_n
    )
  \right),
\end{align}
where $\gamma_{k}(\theta_{E,n},\psi_{E,n},{\cal K}_n)$ 
denotes
the SINR of the $k$th user signal wiretapped by the eavesdropper,
expressed by \cref{eq:gamma_uk}
located at the top of this page.
\begin{table*}[!t]
  \normalsize
  \centering
  \begin{minipage}{\textwidth}
    \begin{align}
      \label{eq:gamma_uk}
        \gamma_k(\theta_{k,n},\psi_{k,n},{\cal K}_n) 
        &=
          \displaystyle
          \frac{
          \overbrace{
            \max_{{\cal K}_{n}}
            \left[
            \left(
              \sum_{
                n \in {\cal I}_k
              }
              h_{k,n}
              \left(
                \theta_{k,n},\psi_{k,n}
                \middle|
                \Psi
              \right)
              \right)^2
              \beta_k({\cal K}_{n})
              \right]
          }^{{\rm Intended~signal~for~the~}{\it k}{\rm th~user~from~}{\it {\cal I}_k}}
          }
          {
          \underbrace{
            \max_{{\cal K}_n}
            \left[
            \left(
            \sum_{
              n \in {\cal I}_k
            }
            h_{k,n}
            \left(
              \theta_{k,n},\psi_{k,n}
              \middle|
              \Psi
            \right)
            \right)^2
            \sum_{\overline{k}={k+1}}^{\left|{\cal K}_{n}\right|}
            \beta_{\overline{k}}({\cal K}_{n})
            \right]
          }_{{\rm Interference~of~NOMA~corresponding~to~the~other~users~in~}{\it {\cal I}_k}}
          +
          \underbrace{
            \left(\sum_{
            \tilde{n}
            \in 
            {\cal I}_{k}^{\tilde n}}
            h_{{k},\tilde{n}}
            \left(
              \theta_{k,\tilde{n}},\psi_{k,\tilde{n}}
              \middle|
              \Psi
            \right)
            \right)^2
          }_{{\rm Interference~from~the~other~LEDs}}
          +
          \frac{\sigma_N^2}{P_s}
          }
    \end{align}
    \medskip
    \hrule
  \end{minipage}
\end{table*}

In this work, 
the \emph{transmission sum rate} among the legitimate users
and the \emph{secrecy sum rate}
are defined as~\cite{pham_2017_secrecy}
\begin{align}
  R_T &= 
  \sum_{n=1}^N\sum_{k=1}^K 
    R_{D_k}\left(\theta_{k,n},\psi_{k,n}\right),\\
  R_S &= \sum_{n=1}^N\sum_{k=1}^K 
  \left[R_{D_k}\left(\theta_{k,n},\psi_{k,n}\right) 
    - 
    R_{E,D_k}\left(\theta_{E,n},\psi_{E,n}\right)
  \right]^+,
  \label{sum-rate}
\end{align}
respectively.
\section{Transmission Strategies and Power Allocation Schemes}
\label{sec:strategies}

In this section, we introduce the three transmission strategies and two power allocation schemes
investigated in this work.

\subsection{Transmission Strategies}

We first introduce the conventional 
{\it broadcasting strategy} used in~\cite{chen_2018_on}.
Then,
we propose two new strategies referred to as
{\it simple and smart LED linking strategies},
depending on the various locating scenarios of users
in order to improve their transmission performance.

\subsubsection{Broadcasting Strategy}

In this transmission strategy,
all the LEDs transmit the same superposed signal 
by NOMA~\cite{chen_2018_on}.
The transmitted signal 
by $S_n$,
denoted by $x_n^{\rm b}$ for this strategy,
is
similar to \cref{eq:x_n}
with 
${\cal K}_n = {\cal K}$,
which
can be expressed as
\begin{align}
  {x}^{\rm b}_n = 
  \sum_{k \in \mathcal{K}}
  \sqrt{
    P_s
    \beta_k
    ({\cal K})
  }
  s_{n,k}
  +
  I_{\mathrm{DC}}.
\end{align}

\subsubsection{Simple LED Linking Strategy}

In this strategy,
the system will transmit signals based on the estimated location 
of legitimate users. 
The LED will transmit signals to users located in its coverage area.
In case there exist more than one user in the coverage of the LED, 
the signal will be transmitted by the multiple access scheme,
i.e.,
NOMA.
Hence, the
transmitted signal by $S_n$,
denoted by
$x_n^{\rm smp}$,
can be expressed as
\begin{align}
  {x}^{\mathrm{smp}}_n = 
  \begin{cases}
    \sum_{k \in \mathcal{K}_n}
    \sqrt{\beta_k({\cal K}_n)
    P_s}
    s_{n,k}
    +
    I_{\mathrm{DC}},
    &
    {\cal K}_n \neq \varnothing,\\
    I_{\mathrm{DC}},
    & \text{otherwise}.
  \end{cases}
\end{align}

A disadvantage of this method is that the devices 
located in the overlapping area 
may be interfered with the signals sent from 
different LED sources, and the interference causes some rate loss in terms of transmission performance. 
In order to overcome such a rate loss, 
we further introduce a novel smart LED linking strategy.

\subsubsection{Smart LED Linking Strategy}

    \begin{algorithm}[t]
      \caption{Smart LED linking strategy}
      \label{algo:Intel}
      \begin{algorithmic}[1]
        \Require{$D_k(x_{D_k},y_{D_k},z_{D_k})$~device coordinates, $S_n(x_{S_n},y_{S_n},Z)$~LED coordinates, 
        $\theta_{1/2}$~half illuminance angle.}
        \Ensure {Group set of LEDs ${\LL}$}
        \State {\textbf{Initialize}: $r \gets (Z - z_{D_k})\tan \theta_{1/2}$, ${\LL} \gets \varnothing$}
        \For {$n = 1:N$}
            \State {${\cal K}_n \gets \varnothing$, ${\cal M} \gets \varnothing$,
            ${\cal L}_n \gets \varnothing$}
            \State {${\LL} \gets {\LL} \cup \{{\cal L}_n\}$}
            \For {$k = 1:N$}
            \If {$(x_{D_k} - x_{S_n})^2 + (y_{D_k} - y_{S_n})^2 \leq r^2$} 
                \State {${\cal K}_n \gets {\cal K}_n \cup \{k\}$}
            \EndIf
            \State {$k \gets k + 1$}
            \EndFor
            \For {$l = 1:n$}
              \If {${\cal L}_l \cap {\cal K}_n \neq \varnothing$} 
                \State {${\cal M} \gets {\cal M} \cup \{l\}$, ${\cal K}_n \gets {\cal L}_l \cup {\cal K}_n$}
              \EndIf
            \EndFor
            \For {$m \in {\cal M}$} 
              \State {${\cal L}_m \gets {\cal K}_n$}
            \EndFor
            \State {${\cal L}_n \gets {\cal K}_n$}
        \EndFor
        \State \Return {${\LL}$}
      \end{algorithmic}
    \end{algorithm}

Since it may be difficult for PDs
to detect the interfered symbol in the simple LED linking strategy,
we propose the smart LED linking strategy
in order to avoid receiving the interfering symbols.
In this strategy,
the system will first determine 
if there is any user in the overlapping coverage area.
If so, 
the system will synchronize the two LEDs covering the overlapping area.
These linked LEDs will transmit the same signal independently,
according to the needs of users in that region.
This strategy
is described by~\cref{algo:Intel},
where we denote the set of group sets of LEDs as $\LL$, 
and the transmitted signal in the $n$th group set ${\cal L}_n$
(${\cal L}_n \in \LL$) 
by $S_n$. 
Similar to the \emph{simple LED linking}, 
the transmitted signal, 
denoted by $x_n^{\rm smt}$, 
is expressed as
\begin{align}
  x^{\mathrm{smt}}_n = 
  \begin{cases}
    \sum_{l \in {\cal L}_n}
  \sqrt{\beta_l
  ({\cal L}_n)
  P_s}
  s_{n,l}
  +
  I_{\mathrm{DC}},
  &
  {\cal L}_n \neq \varnothing,\\
  I_{\mathrm{DC}},
  & \text{otherwise}.
  \end{cases}
\end{align}

\subsection{Power Allocation Schemes Based on Maximum Sum Rate}

We introduce two distinct power allocation schemes within the context of NOMA in this work.
The first scheme is the conventional 
\emph{fixed power allocation scheme}~\cite{dogra_2022_journal}.
The second, 
as our novel contribution,
is the
\emph{maximum sum rate-based scheme},
which optimizes power distribution by estimating the maximum sum rate, thereby significantly enhancing the performance and efficiency of NOMA systems.

\subsubsection{Fixed Power Allocation Scheme}
In this scheme, we assume that the power ratio 
among different users
is fixed.
In SIC, 
the receiver will detect the signal with stronger power first.
Therefore,
one of the power ratio parameters of NOMA
$\beta_k^{\mathrm{fix}}({\cal K}_n)$
for $k \in {\cal K}_n$
can be defined as~\cite{dogra_2022_journal}
\begin{align}
  \beta_k^{\mathrm{fix}}({\cal K}_n)
  \triangleq
  \begin{cases}
    \zeta(1 - \zeta)^{k-1}, &~\mbox{if}~~k < |{\cal K}_n|, \\
    (1 - \zeta)^{k-1}, &~\mbox{if}~~k = |{\cal K}_n|,
  \end{cases}
\end{align}
where
$\zeta \in (0.5,1]$ is the fixed power ratio.

\subsubsection{Maximum Sum Rate-based Scheme}

In this scheme, 
the LED is able to optimize the power allocation in order to maximize the sum rate only 
by the positions of users.
By defining the power ratio parameter of NOMA in this scheme as 
$\beta_k^{\mathrm{opt}}({\cal K}_n)$,
the corresponding optimization problem 
of maximizing the sum of the optimized transmission rate
$\widehat{R}_{k}$ 
is formulated as follows~\cite{boyd_2004_convex}:
\begin{subequations}
  \label{eq:optimization}
  \begin{align}
    \text{maximize} 
    \quad & \sum_{k \in \mathcal{K}}\widehat{R}_{k} 
    \left(\overline{\lambda}_{k},\tilde{\omega}_{k,n},\beta_k^{\mathrm{opt}}({\cal K}_n)\right) 
    \tag{\ref{eq:optimization}}\\
    \mathrm{subject~to} \quad & \sum_{k \in \mathcal{K}} \beta_k^{\mathrm{opt}}({\cal K}_n) \leq 1, \\
    & \beta_{\hat{k}}^{\mathrm{opt}}({\cal K}_n) \geq 
      \beta_{\hat{k}+1}^{\mathrm{opt}}({\cal K}_n), 
      \\
    & \beta_k^{\mathrm{opt}}({\cal K}_n) \geq 0
  \end{align}
\end{subequations}
for $\hat{k} \in \{1,2,\ldots,K-1\}$.
In this optimization problem,
the radiation angle 
$\overline{\lambda}_{k}$ 
is set as the expected value 
of $\lambda_k$ defined in~\cref{eq:omega}
and
the azimuth angle 
$\tilde{\omega}_{k,n}$ of $D_k$ is set as the argument of the maximum sum rate.
As a result,
the formula of the optimization problem can be rewritten as
\begin{align}
  &\sum_{k \in \mathcal{K}}\widehat{R}_{k} \left(\overline{\lambda}_{k},\tilde{\omega}_{k,n},\beta_k^{\mathrm{opt}}({\cal K}_n)\right)
  \nonumber\\
  &= \sum_{k \in \mathcal{K}}\frac{1}{2} \log_2 \left(
    1 + 
    \widehat{\gamma}_{k}\left(
      \overline{\lambda}_{k},\tilde{\omega}_{k,n},\beta_k^{\mathrm{opt}}({\cal K}_n),
      {\cal K}_n
      \right)
    \right),
\end{align}
where
$\widehat{\gamma}_{k}
\left(
  \overline{\lambda}_{k},
  \tilde{\omega}_{k,n},
  \beta_k^{\mathrm{opt}}({\cal K}_n),
  {\cal K}_n
\right)$ is 
defined in \cref{eq:widehatgamma_uk}
located at the top of the next page
\begin{table*}
  \centering
  \normalsize
  \begin{minipage}{\textwidth}
    \begin{align}
      &\widehat{\gamma}_{k}
      \left(
        \overline{\lambda}_{k},
        \tilde{\omega}_{k,n},
        \beta_k^{\mathrm{opt}}({\cal K}_n),
        {\cal K}_n
      \right)
      \nonumber\\
      &=
      \left\{
      \displaystyle
      \overbrace{
        \max_{{\cal K}_{n}}
        \left[
          \left(
            \sum_{n\in {\cal I}_k}
            \hat{h}_{{k},n}
              \left(
                \theta_{k, n}, 
                \hat{\psi}_{k, n}(\overline{\lambda}_{k},\tilde{\omega}_{k,n})
                \middle|
                \Psi
              \right)
          \right)^2
        \beta_k({\cal K}_{n})
        \right]
      }^{{\rm Decoded~information~of~}{\it k}{\rm th~user~from~}{\it {\cal I}_k}}
      \right\}
      \nonumber\\
      &\cdot
      \left\{
        \displaystyle
        \underbrace{
          \max_{{\cal K}_{n}}
          \left[
          \left(
            \sum_{n \in {\cal I}_k}
            \hat{h}_{{k},n}
              \left(
                \theta_{k, n}, 
                \hat{\psi}_{k, n}(\overline{\lambda}_{k},\tilde{\omega}_{k,n})
                \middle|
                \Psi
              \right)
          \right)^2
          \sum_{\overline{k}={k+1}}^{\left|{\cal K}_{n}\right|}
          \beta_{\overline{k}}({\cal K}_{n})
          \right]
          }_{{\rm Interference~of~NOMA~from~other~users~in~}{\it {\cal I}_k}}
        +
        \underbrace{
          \left(\sum_{
            \tilde{n}
            \in 
            {\cal I}_{k}^{\tilde n}
            }
            \hat{h}_{{k},\tilde{n}}
            \left(
              \theta_{k, \tilde{n}}, 
              \hat{\psi}_{k, \tilde{n}}(\overline{\lambda}_{k},\tilde{\omega}_{k,\tilde{n}})
              \middle|
              \Psi
            \right)
          \right)^2
        }_{{\rm Interference~from~the~other~LEDs}}
        +
        \frac{\sigma_N^2}{P_s}
      \right\}^{-1}
      \label{eq:widehatgamma_uk}
    \end{align}
    \medskip
    \hrule
  \end{minipage}
\end{table*}
with the estimated channel gain
$\hat{h}_{{k},n}
\left(
  \theta_{k, n}, 
  \hat{\psi}_{k, n}(\overline{\lambda}_{k},\tilde{\omega}_{k,n})
  \middle|
  \Psi
\right)$
formulated as
\begin{align}
  \label{eq:op_h_uk}
  \hat{h}_{{k},n}&
  \left(
    \theta_{k, n}, 
    \hat{\psi}_{k, n}(\overline{\lambda}_{k},\tilde{\omega}_{k,n})
    \middle|
    \Psi
  \right) \nonumber\\
  &=
  A \frac{(m+1) R_{\mathrm{PD}}}{2 \pi} 
  \cos^m\left({\theta}_{k, n}\right)
  \frac{
    \cos \left(\hat{\psi}_{k, n}(\overline{\lambda}_{k}, 
    \tilde{\omega}_{k, n})\right)
  }
  {d_{D_k, n}^2}
  \nonumber\\
  & \quad \quad \cdot g\left(\hat{\psi}_{k, n}(\overline{\lambda}_{k}, \tilde{\omega}_{k, n}) \middle| \Psi\right),\\
  \label{eq:op_theta_uk}
  \overline{\lambda}_{k} 
  &= 
  \EE\left[{\lambda_{k}}\right],\\
  \label{eq:op_theta_omega}
  \tilde{\omega}_{k, n} 
  &= 
  \mathop{\arg \max}_{\omega_{k,n} \in [0,2\pi)}
  \left(
    \hat{h}_{{k},n}
    \left(
      \theta_{k, n}, 
      \hat{\psi}_{k, n}(\overline{\lambda}_{k},\omega_{k,n})
      \middle|
      \Psi
    \right)
  \right).
\end{align}
Note that
$\hat{\psi}_{k, n}(\overline{\lambda}_{k}, \tilde{\omega}_{k, n})$
can be derived by substituting
\cref{eq:op_theta_uk,eq:op_theta_omega} into \cref{eq:omega}.
Even though~\cref{eq:optimization} is a non-linear constrained optimization problem,
one can solve it numerically by 
JuMP package with Ipopt optimizer~\cite{dunning_2017_jump,wachter_2006_ipopt}\footnote{
  The case of the indicator function
  introduces non-convexity into the
  optimization problem
  defined in~\cref{eq:optimization},
  specifically with regard to the channel gain
  $\hat{h}_{k,n}$.
  This non-convex nature poses
  a challenge for applying
  convex optimization or approximation methods
  to address the problem.
  However,
  it is important
  to note that the focus of
  the current study does not 
  extend to exploring
  convex optimization for this particular challenge.
  Instead,
  our interest lies in establishing
  the maximum on the transmission
  and secrecy performance.
  Despite the drawback of high complexity,
  the chosen optimizer
  is sufficient for achieving our objectives.
}.

\section{Simulation Results and Discussion}
\label{sec:simu}

In this section, we first describe the setups of 
the numerical environment in our proposed VLC system model
with the mathematical body blockage.
Additionally, 
we evaluate the transmission sum rate of the users 
as well as the secrecy sum rate.
We also investigate the performance difference among transmission strategies 
and the power allocation schemes of NOMA.

\subsection{Parameter Settings for Simulations}

\begin{table}[!t]
    \centering
    \small
    \caption{Simulation parameters}
    \label{tab:simu}
    \begin{tabularx}{\columnwidth}{X|c}
        \noalign{\hrule height 1pt}
        \multicolumn{2}{c}{\textbf{Room configuration}}\\ \noalign{\hrule height 1pt}
        Length $L$ $\times$ Width $W$ & 40 $\times$ 40~$\mathrm{m}^2$ \\ 
        Height of the room $Z$ & 3.98~$\mathrm{m}$ \\
        Number of LED arrays & 23 \\
        Number of users $K$ & 6 \\
        The side length of the equilateral triangle of LED sources $l$
        & $9.6~$m \\
        \noalign{\hrule height 1pt}
        \multicolumn{2}{c}{\textbf{LED electrical and optical characteristics}} \\ 
        \noalign{\hrule height 1pt}
        Average optical power per LED & 0.25~W \\
        Half-intensity angle $\theta_{\frac{1}{2}}$ & $70^{\circ}$ \\
        Power ratio of NOMA $\zeta$ & 0.6 \\
        \noalign{\hrule height 1pt}
        \multicolumn{2}{c}{\textbf{Photodiode characteristics}} \\
        \noalign{\hrule height 1pt}
        Received fiend $\Psi$ & $60^{\circ}$ \\
        Physical area $A$ & 1~$\text{cm}^2$ \\ 
        Noise power $\sigma_N^2$ & -98.35~dBm \\
        Reflective index $\eta$ & 1.5 \\
        Height ratio parameter of the device $\nu$ & 0.75 \\
        \noalign{\hrule height 1pt}
        \multicolumn{2}{c}{\bf \makecell{Parameters related to \\the bodies of users and eavesdropper}} \\
        \noalign{\hrule height 1pt}
        Mean of the height of the users and eavesdropper $H$ & 1.6~$\mathrm{m}$ \\
        Standard deviation of the height of the users and the eavesdropper $\sigma_H$ & 0.1333~$\mathrm{m}$ \\
        Distance from the center of the top base of cylinder to the device
        $l_d$ & 0.4~m\\
        Radius of bodies 
        $r$ & 0.2~m \\
        \noalign{\hrule height 1pt}
    \end{tabularx}
\end{table}

\begin{algorithm}[tb]
    \caption{Random height generation algorithm.}
    \label{algo:random_height}
    \begin{algorithmic}[1]
      \Require{The average body height $H$ and its standard deviation $\sigma_H$.}
      \Ensure {${\sf h}_{D_i}$ with $i \in {\cal K}'$.}
      \ForAll{$i \in {\cal K}'$}
      \While {True}
        \State {Generate ${\sf h}_{D_i}$ with ${\sf h}_{D_i} \sim {\cal N}(H, \sigma_H^2)$.}
        \If {${\sf h}_{D_i}$ > 0}
          \State {\Return ${\sf h}_{D_i}$.}
        \EndIf
      \EndWhile
      \EndFor
    \end{algorithmic}
  \end{algorithm}

\begin{figure}[!ht]
    \centering
    \includegraphics[clip,width=.65\columnwidth, bb=0 0 504 545]{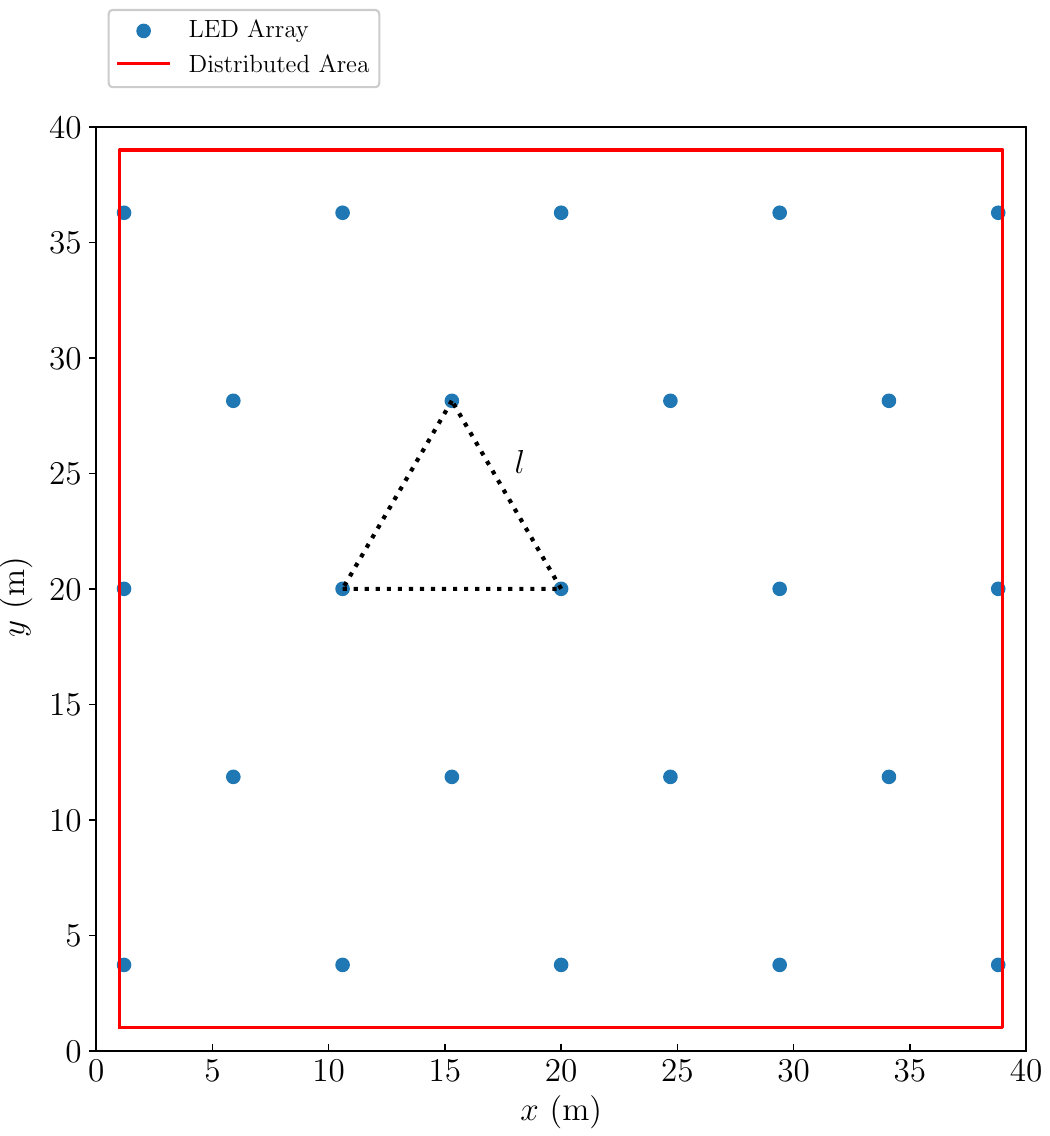}
    \caption{The arrangement of LED sources
    (blue dots represent the positions of LEDs and
        unit of horizontal and vertical axes are measured by meters).}
    \label{fig:led_dist}
\end{figure}

\begin{table}[!ht]
    \centering 
    \small
    \caption{The Cartesian XOY coordinates of users.}
    \label{tab:coords}
    \begin{tabular}{|c||c||c||c|}
        \hline
       & Scenario 1 & Scenario 2 & Scenario 3 \\ \hline\hline
    $D_1$ & $(6,6)$ & $(13, 16)$ & $(10.6, 14.5)$ \\ \hline\hline
    $D_2$ & $(34, 6)$ & $(20, 12)$ & $(15.3, 22.7)$ \\ \hline\hline
    $D_3$ & $(34, 34)$ & $(27, 16)$ & $(5.9, 22.7)$ \\ \hline\hline
    $D_4$ & $(6, 34)$ & $(27, 24)$ & $(34.1, 20)$ \\ \hline\hline
    $D_5$ & $(20, 10)$ & $(20, 28)$ & $(34.1, 32.2)$ \\ \hline\hline
    $D_6$ & $(20, 30)$ & $(13, 24)$ & $(34.1, 7.8)$ \\ \hline
    \end{tabular}
\end{table}

\begin{figure}[!t]
    \centering
    \subfigure[Scenario 1]{
        \includegraphics[clip,width=.8\columnwidth, bb=0 0 451 424]{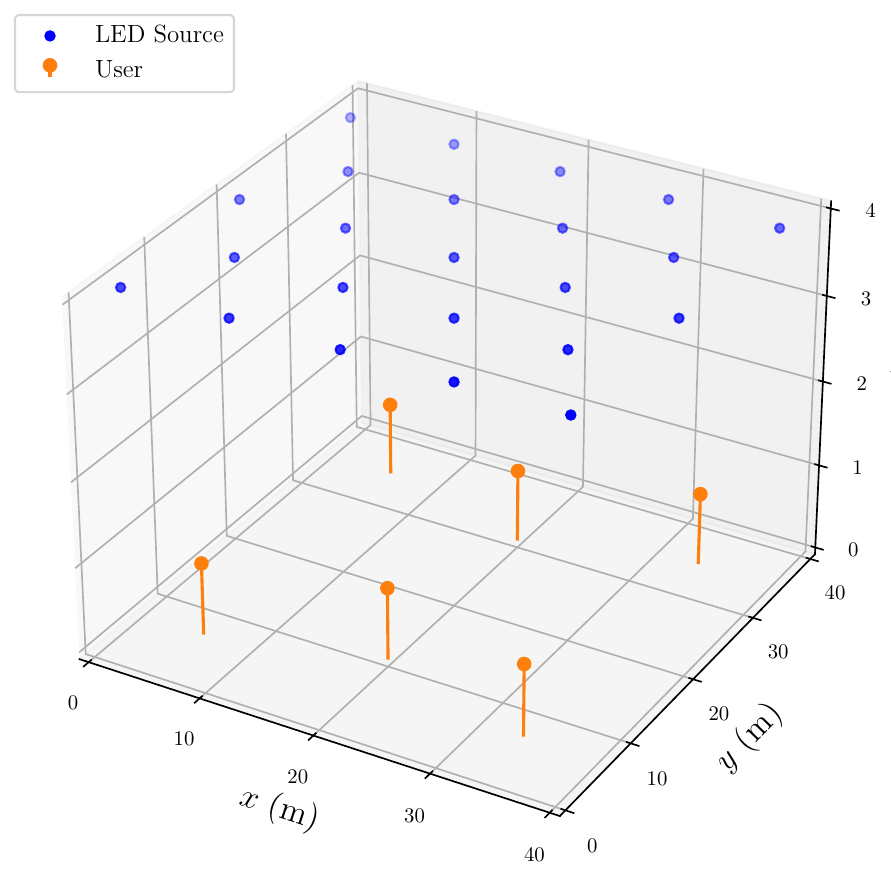}
        \label{fig:user_location_1}
    }
    \subfigure[Scenario 2]{
        \includegraphics[clip,width=.8\columnwidth, bb=0 0 451 424]{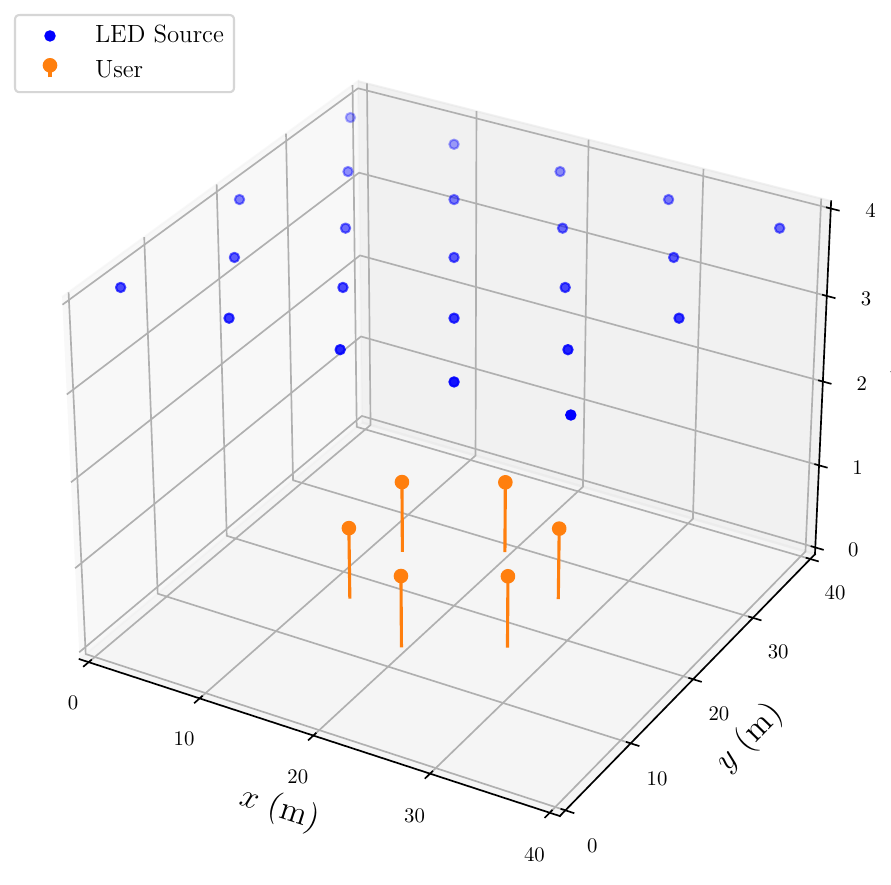}
        \label{fig:user_location_2}
    }
    \subfigure[Scenario 3]{
        \includegraphics[clip,width=.8\columnwidth, bb=0 0 451 424]{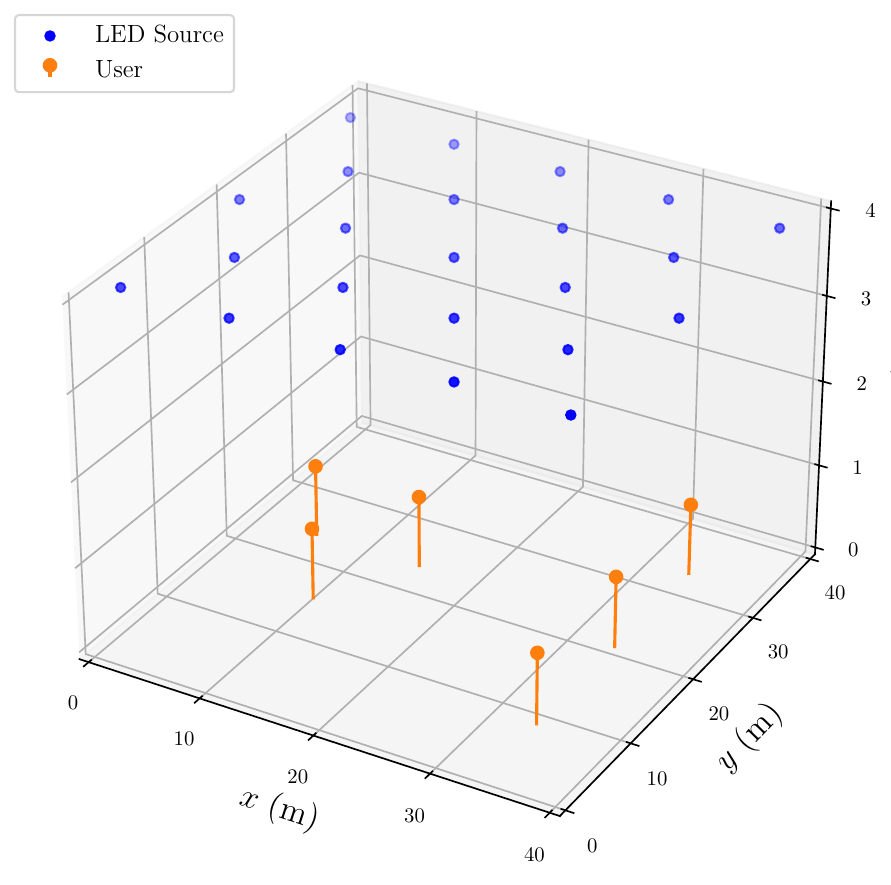}
        \label{fig:user_location_3}
    }
    \caption{Illustration of the location of PDs.}
    \label{fig:user_location}
\end{figure}

The parameters for the numerical simulations are listed in~\cref{tab:simu}, 
including LED and PD settings, as well as the body blockages~\cite{soltani_2019_modeling,zhang_2016_grouping}.
Furthermore,
in order to reduce the repetitive coverage
by different LED units, we consider a \emph{new arrangement of LEDs} where each LED is placed in the vertexes of equilateral triangles.
This arrangement is illustrated in~\cref{fig:led_dist}.
The side length of the triangle satisfies the following inequality:
\begin{align}
  \frac{l}{2}\frac{1}{\cos\frac{\pi}{6}} \leq (z - z_{D_k})\tan \theta_{1/2}.
\end{align}
Therefore, we have
\begin{align}
  l \leq \sqrt{3}(z - z_{D_k})\tan \theta_{1/2}.
\end{align}
We set the location of the first LED at~$(20,20)$ in the ceiling,
and 
the remaining LEDs are assigned to the blue points (cf.~\cref{fig:led_dist}), 
where the coordinates of any three adjacent LEDs form an equilateral triangle with side length $l$.
The geometric coordinates of users are set as
in~\cref{tab:coords} with variance $z_{D_k}$ in each simulation loop
(also illustrated in~\cref{fig:user_location});
Scenario~1 corresponds to the case where all the users are 
sparsely located in the room (cf.~\cref{fig:user_location_1}),
Scenario~2 corresponds to the case where all the users are 
closely located in the room, 
and finally (cf.~\cref{fig:user_location_2}),
Scenario~3 corresponds to the case of a hybrid combination of the two scenarios
(cf.~\cref{fig:user_location_3}).
The Monte Carlo simulation method is employed, with $10^4$
iterations determined 
by the trade-off between simulation accuracy and computational complexity.
Also,
in this work,
the height of each user ${\sf h}_{D_{k}}$
as well as 
the eavesdropper ${\sf h}_{D_E}$
is generated by the
\cref{algo:random_height}.
Therefore, 
in each simulation loop,
we first generate the heights of each user and the eavesdropper.
Afterward, 
we calculate the transmission rate for that loop.

\subsection{Transmission Performance}

\begin{figure}[!t]
    \centering
    \includegraphics[clip,width=\columnwidth, bb=0 0 503 387]{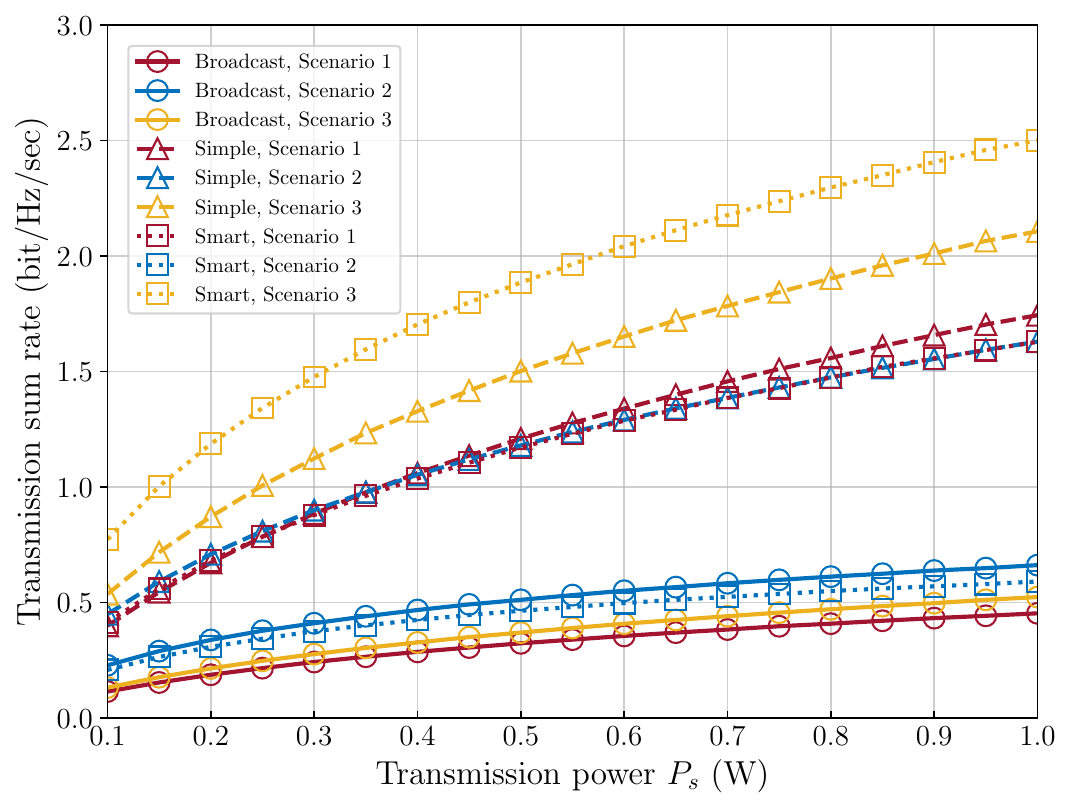}
    \caption{Transmission power versus transmission sum rate with fixed power allocation scheme while the coordinates of the eavesdropper are uniformly distributed
    (solid lines with circle markers are the graphs of the broadcasting strategy,
    dash lines with triangle markers are the graphs of the simple LED linking strategy,
    and 
    dot lines with square markers are the graphs of the smart LED linking strategy).}
    \label{fig:sum_rate_2d}
\end{figure}


In~\cref{fig:sum_rate_2d},
we investigate the transmission sum rate of the three aforementioned transmission scenarios with the proposed strategies and fixed power allocation scheme 
when $x_E$ and $y_E$ are uniformly distributed from 1~m to 39~m
(i.e., in the distributed area),
where the cylinder bodies may block the lights received by the users.
From this figure, 
we can observe that our proposed simple LED linking and smart LED linking strategies outperform the conventional broadcasting strategy.
For a comparison of the performance of 
the two LED linking strategies, 
the results
depend on how users are distributed.
In Scenarios~1~and~2,
the simple LED linking strategy outperforms
the smart LED linking strategy,
whereas 
the smart LED linking strategy 
outperforms 
the simple LED linking strategy in Scenario~3.
From Fig.~\ref{fig:sum_rate_2d},
we also observe that 
this performance difference is not affected by
the transmission power due to the fact that
the location of the users is fixed.
Apparently, 
the distribution of the users affects the transmission performance of the system.
\begin{figure}[t]
    \centering
    \includegraphics[clip,width=\columnwidth]{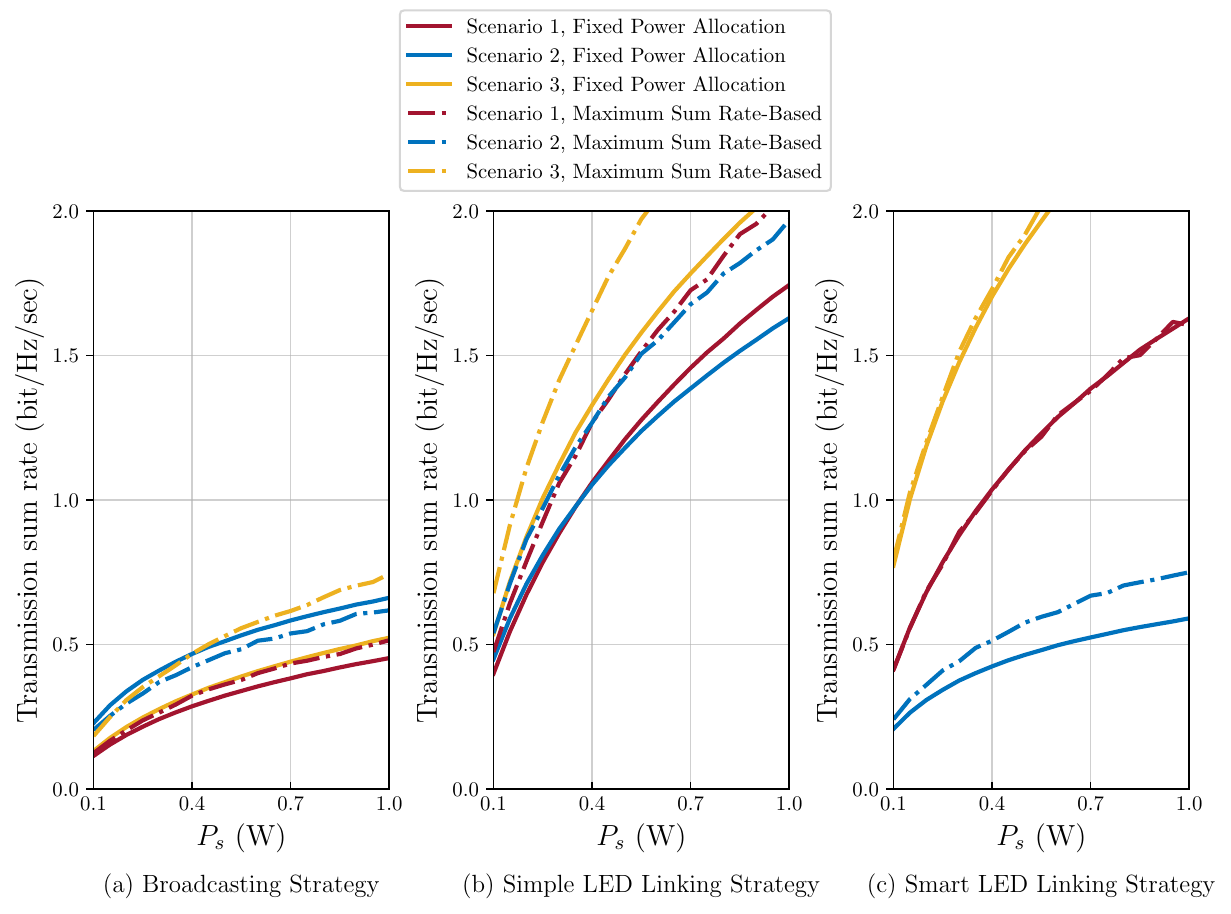}
    \caption{Transmission power versus transmission sum rate with the fixed power allocation scheme and maximum sum rate-based scheme.}
    \label{fig:sum_rate_opt}
\end{figure}

We next analyze the transmission performance of our model under the three strategies, 
considering both the maximum sum rate-based scheme and the fixed power allocation scheme.
The simulation results are depicted in~\cref{fig:sum_rate_opt}.
From the figure,
we can observe that
in~\cref{fig:sum_rate_opt}(a), 
the transmission performance
in
the broadcasting strategy
for Scenario~2 is poor, and this is due to the fact
that the maximum sum rate-based scheme
is only based on the expected value of 
the radiation and azimuth angles, and
the power allocation calculation process will 
cause severe errors when
the users are densely located.
Besides,
the maximum sum rate-based scheme works well 
with the two LED linking strategies.
As shown in
Figs.~\ref{fig:sum_rate_opt}(b)~and~\ref{fig:sum_rate_opt}(c), the values of the transmission sum rates for this scheme are basically higher than that for the fixed power allocation scheme
based on the sum rate maximization of the legitimate users.
Furthermore,
we can observe that in~Fig.~\ref{fig:sum_rate_opt}(c),
the maximum sum rate-based scheme in
Scenario~1 and Scenario~3
does not significantly improve 
the performance
due to the fact that
the locations of users are sparser than those in Scenario~2,
i.e.,
each ${\cal L}_n$ contains only a single or two users.
From these observations,
we may conclude that under a realistic VLC system where the light path could be blocked, the newly proposed strategies
(simple and smart LED linking strategies) 
provide better gain for the transmission sum rate compared to the broadcasting strategy.
\begin{figure}[t]
    \centering
    \includegraphics[clip,width=\columnwidth, bb=0 0 584 438]{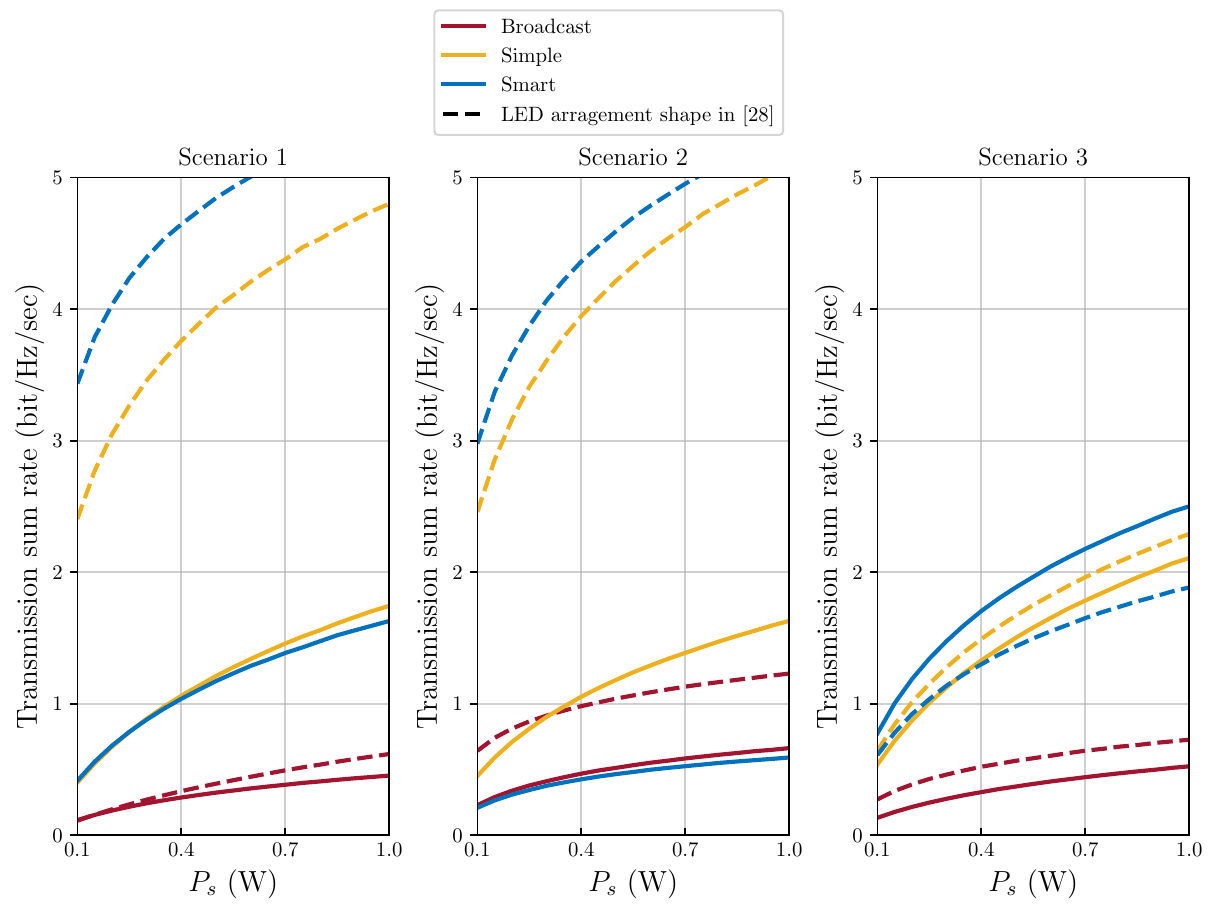}
    \caption{Transmission power versus transmission sum rate with fixed power allocation scheme in the proposed LED arrangement and the LED arrangement in \cite{zhang_2016_grouping} with the same total transmission power.}
    \label{fig:arrangement}
\end{figure}

Additionally,
in~\cref{fig:arrangement},
we compare the transmission sum rate between our proposed LED arrangement and the squared LED arrangement~\cite{zhang_2016_grouping} with the same sidelength 
and total transmission power.
From this figure, we observe that in Scenario~3, 
our proposed scenario outperforms the LED arrangement 
in~\cite{zhang_2016_grouping},
but in Scenarios~1~and~2, 
the transmission sum rate for our proposed arrangement is inferior to that of the arrangement method in~\cite{zhang_2016_grouping}.
This can be attributed to the fact that the squared LED arrangement increases the overlapping area of two geometrically adjacent LEDs, resulting in an increased illuminance for each user. 
If the two adjacent different sources transmit the same symbol simultaneously, 
the transmission performance can be improved.
It is worth noting that the transmission performance of the LED arrangements is also significantly influenced by the locations of users, 
which requires further investigation.

\begin{figure}[!t]
    \centering
    \includegraphics[clip,width=\columnwidth, bb=0 0 503 385]{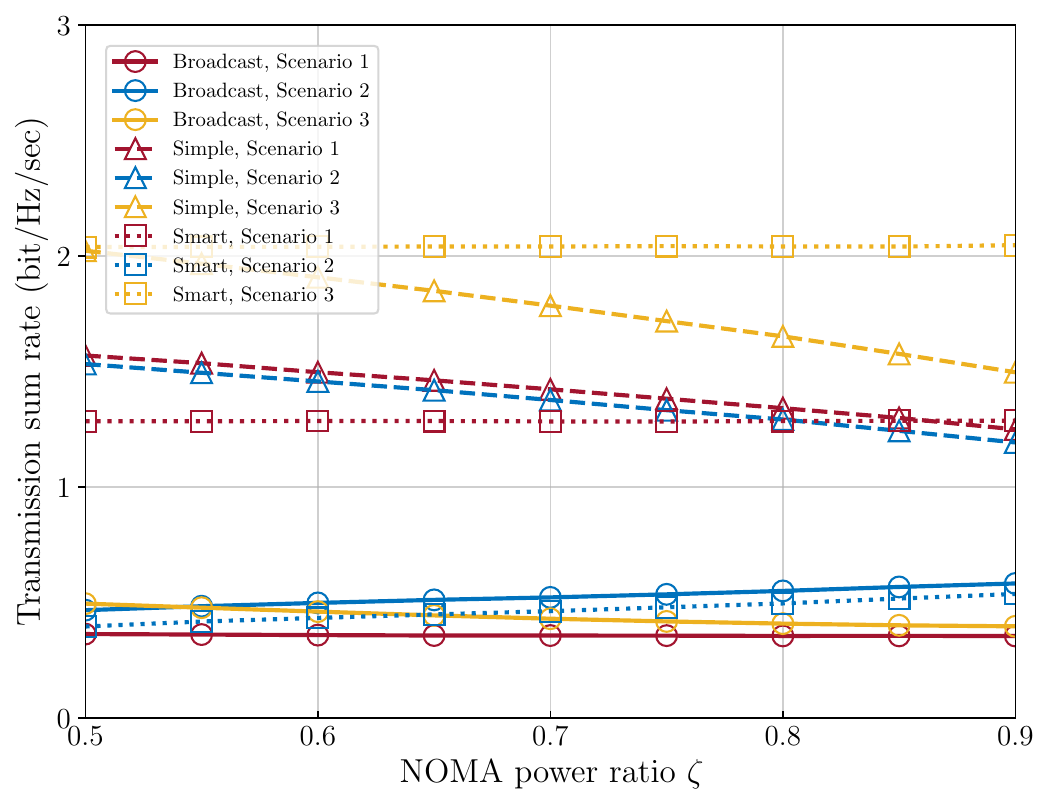}
    \caption{NOMA power ratio $\zeta$ versus transmission sum rate with different LED arrangements and transmission power $P_s = 0.6$~W.}
    \label{fig:zeta}
\end{figure}

Conversely, to observe the effects of the NOMA power ratio $\zeta$, the transmission performance is plotted with respect to $\zeta$ in Fig.~\ref{fig:zeta}.
Notably, the transmission sum rates for the broadcasting strategy and smart LED linking strategy remain constant,
while the simple LED linking strategy experiences a significant decrease as $\zeta$ increases.
The reason for this phenomenon is that,
in the simple LED linking strategy,
users may receive signals from two different adjacent sources based on their location.
The power of the interference decreases with the increase of $\zeta$.

\subsection{Secrecy Performance}

\begin{figure}[!ht]
    \centering
    \subfigure[Scenario~1]{
        \includegraphics[clip,width=.75\columnwidth, bb=0 0 453 424]{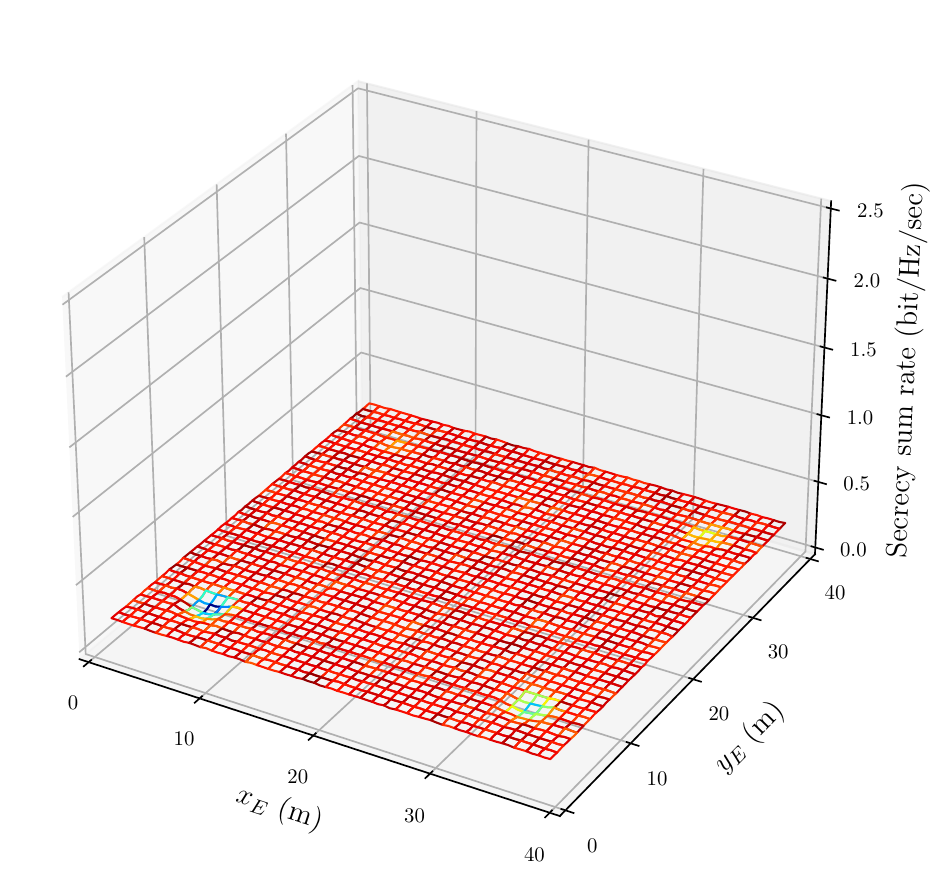}
        \label{fig:sec_3d_1_1}
    }
    \subfigure[Scenario~2]{
        \includegraphics[clip,width=.75\columnwidth, bb=0 0 453 424]{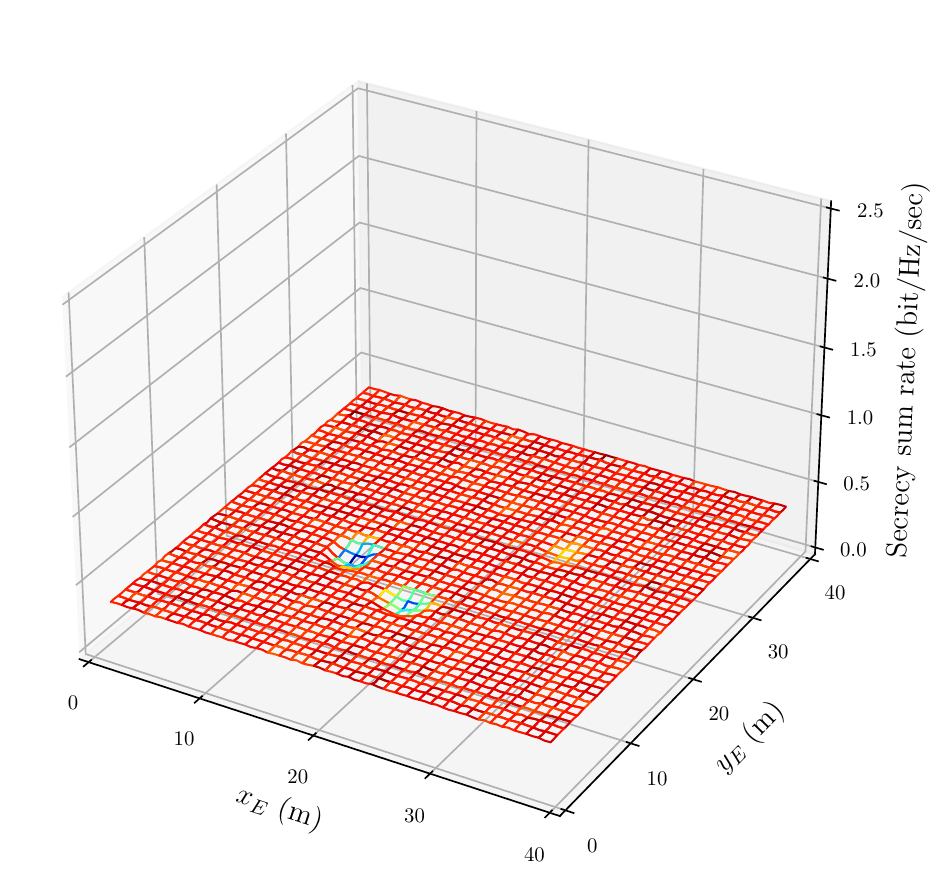}
        \label{fig:sec_3d_1_2}
    }
    \subfigure[Scenario~3]{
        \includegraphics[clip,width=.75\columnwidth, bb=0 0 453 424]{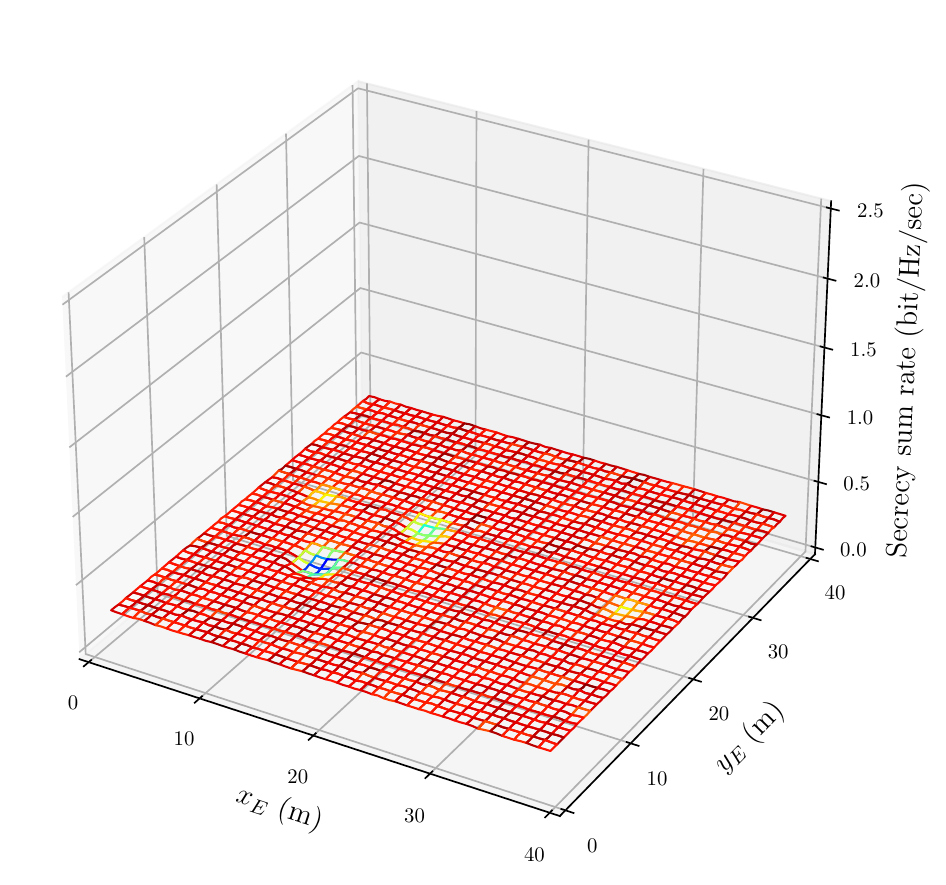}
        \label{fig:sec_3d_1_3}
    }
    \caption{Three-dimensional plots of XOY location of the eavesdropper versus secrecy sum rate
    with the fixed power allocation scheme in the broadcasting strategy.}
    \label{fig:sec_3d_1}
\end{figure}
\begin{figure}[!ht]
    \centering
    \subfigure[Scenario~1]{
        \includegraphics[clip,width=.75\columnwidth, bb=0 0 453 424]{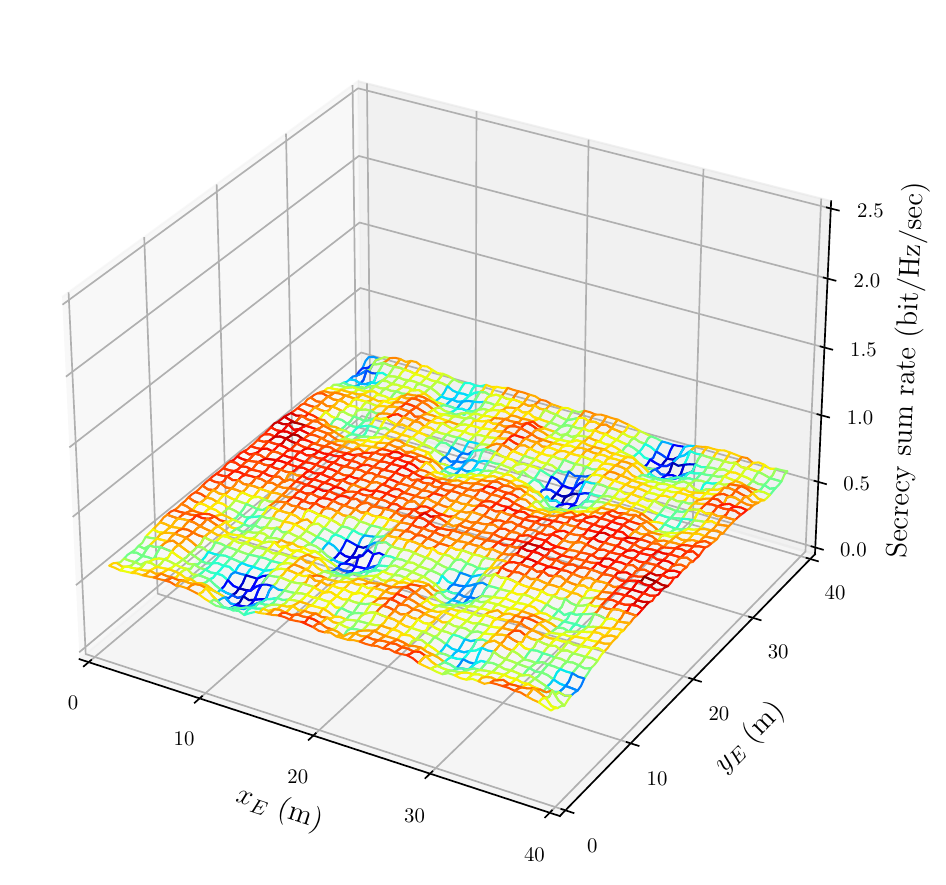}
        \label{fig:sec_3d_2_1}
    }
    \subfigure[Scenario~2]{
        \includegraphics[clip,width=.75\columnwidth, bb=0 0 453 424]{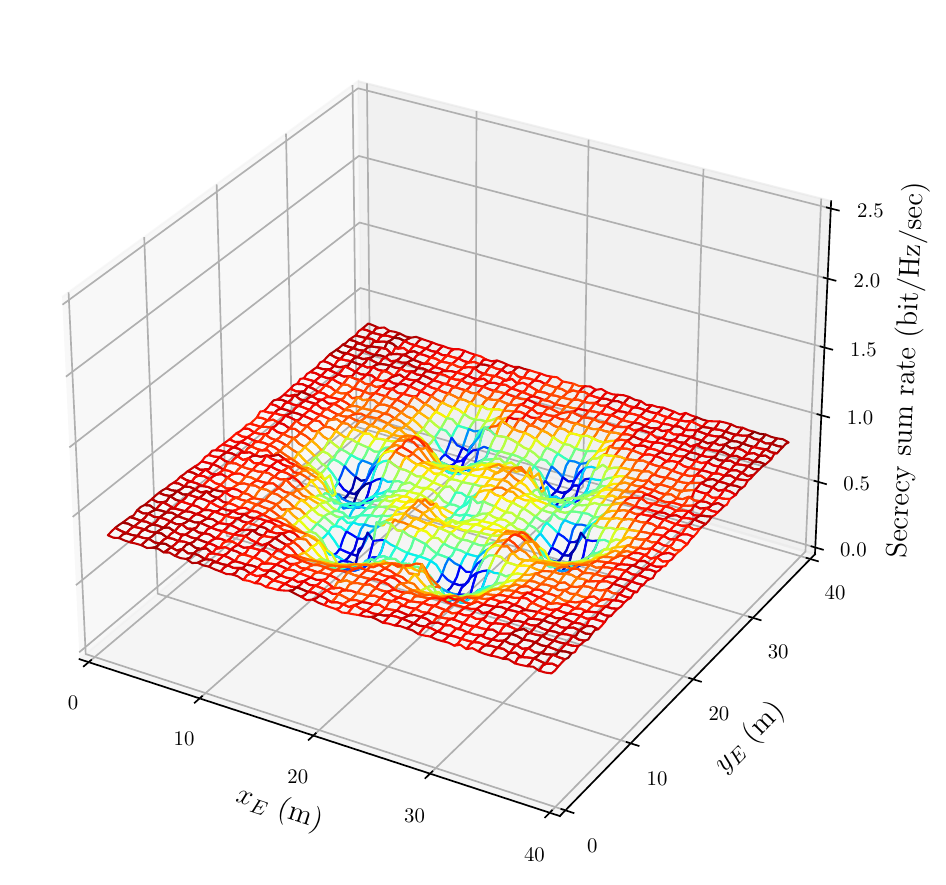}
        \label{fig:sec_3d_2_2}
    }
    \subfigure[Scenario~3]{
        \includegraphics[clip,width=.75\columnwidth, bb=0 0 453 424]{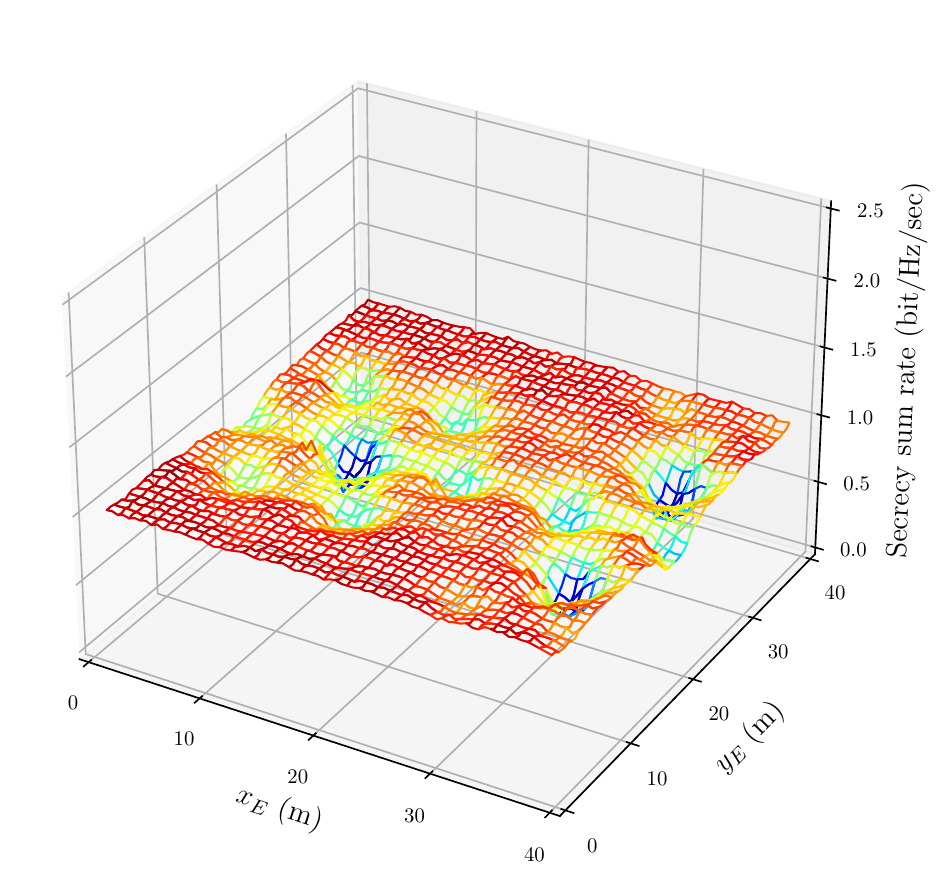}
        \label{fig:sec_3d_2_3}
    }
    \caption{Three-dimensional plots of XOY location of the eavesdropper versus secrecy sum rate
    with the fixed power allocation scheme 
    in the simple LED linking strategy.}
    \label{fig:sec_3d_2}
\end{figure}
\begin{figure}[!ht]
    \centering
    \subfigure[Scenario~1]{
        \includegraphics[clip,width=.75\columnwidth, bb=0 0 453 424]{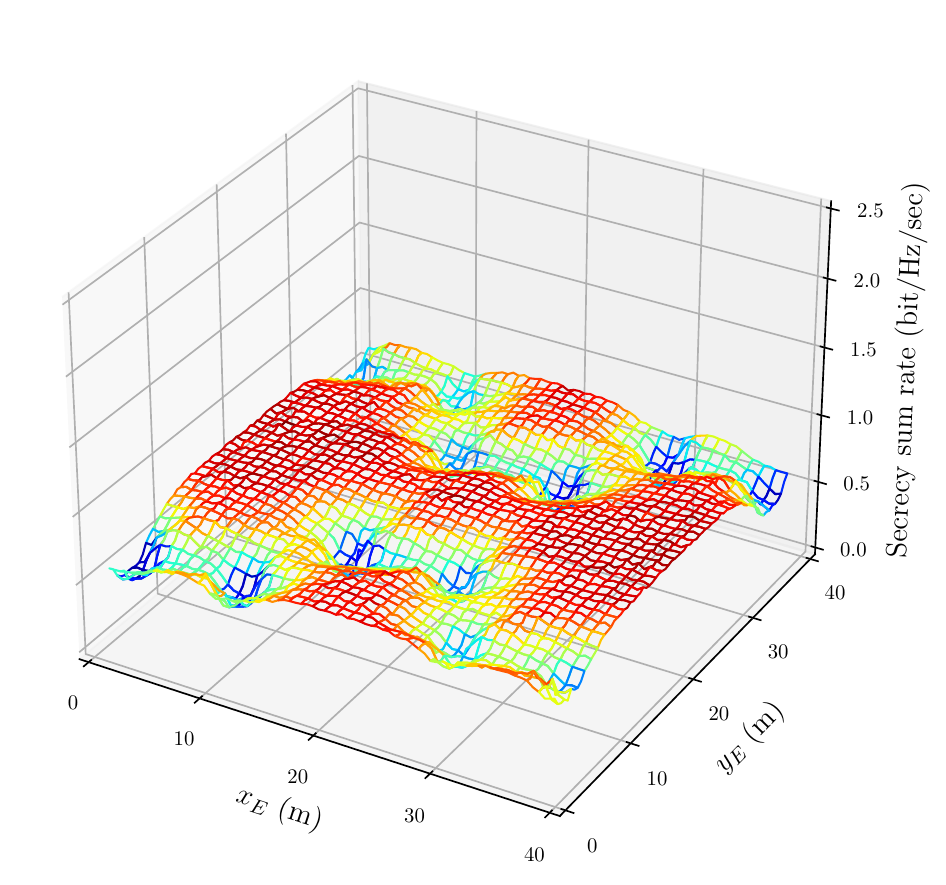}
        \label{fig:sec_3d_3_1}
    }
    \subfigure[Scenario~2]{
        \includegraphics[clip,width=.75\columnwidth, bb=0 0 453 424]{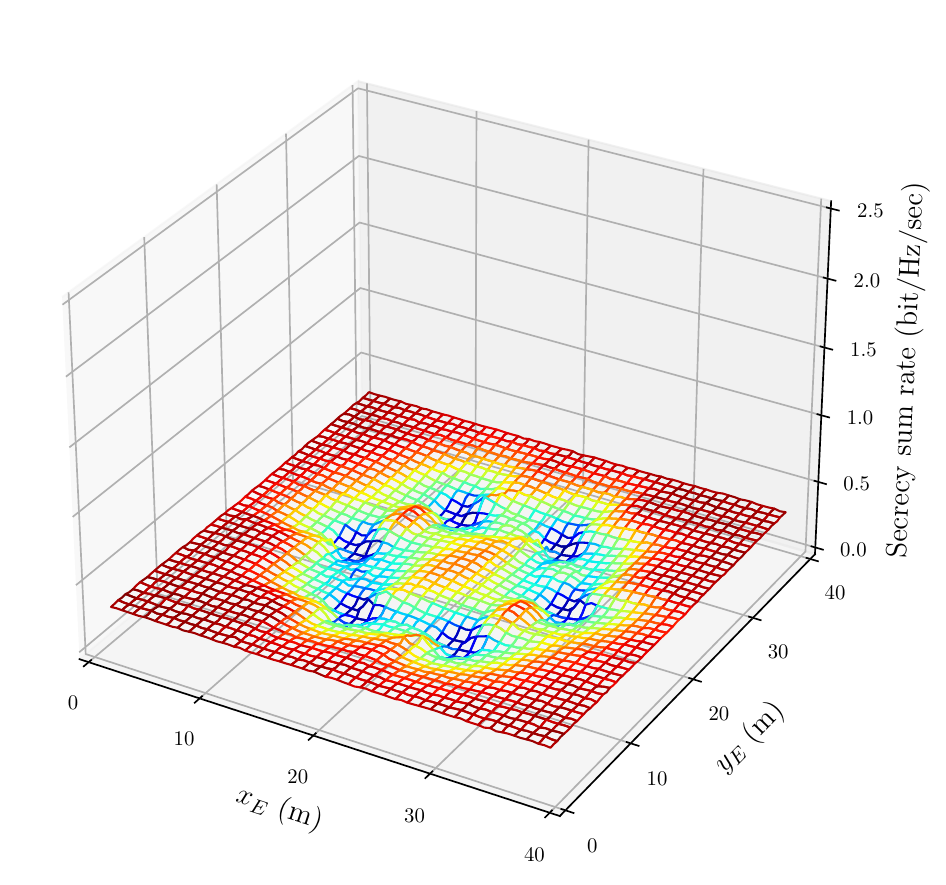}
        \label{fig:sec_3d_3_2}
    }
    \subfigure[Scenario~3]{
        \includegraphics[clip,width=.75\columnwidth, bb=0 0 453 424]{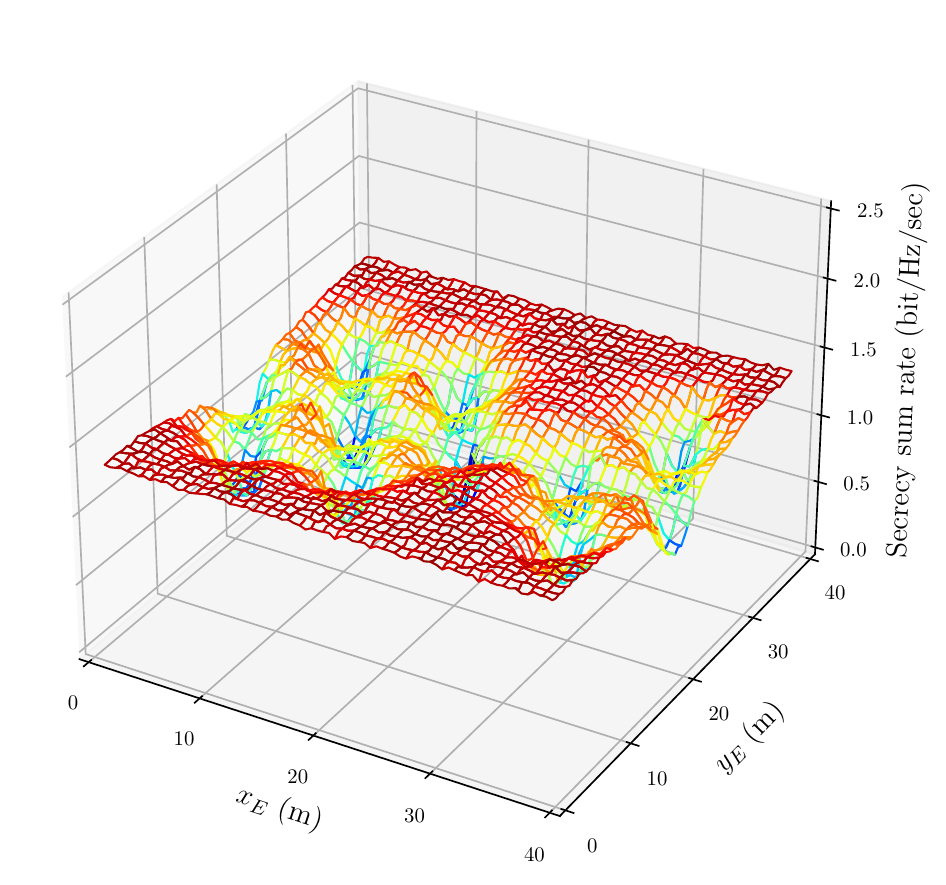}
        \label{fig:sec_3d_3_3}
    }
    \caption{Three-dimensional plots of XOY location of the eavesdropper versus secrecy sum rate
    with the fixed power allocation scheme
    in the smart LED linking strategy.}
    \label{fig:sec_3d_3}
\end{figure}

The secrecy sum rate and the location of the eavesdropper of our system model are examined with the three transmission strategies under the fixed power allocation scheme. 
The simulation results for the broadcasting, simple LED linking, and smart LED linking strategies are shown in Figs.~\ref{fig:sec_3d_1}, \ref{fig:sec_3d_2}, and \ref{fig:sec_3d_3},
respectively.
From these figures,
it is evident that the broadcasting strategy has the lowest secrecy sum rate.
The secrecy performances achieved by the simple as well as 
smart LED linking strategies rely on the distribution of the users.
More precisely, 
for Scenario~2, 
corresponding to the case 
where the users are in a small area with high concentration, 
the simple LED linking strategy provides better secrecy performance 
than the smart one 
(c.f.,
Figs.~\ref{fig:sec_3d_2_2}~and~\ref{fig:sec_3d_3_2}). 
The reason for this behavior 
is that the LEDs 
have the grouping process 
before transmitting data to users
in the smart LED linking strategy in order to bypass signal interference 
in the overlapping regions.
When all the users are 
densely located,
they may be assigned to the same group. 
Therefore, only a subset of the entire LEDs is selected to transmit signals to the legitimate users, 
and the remaining LEDs
make
no contribution to increasing the secrecy sum rate.
In different circumstances,
if some users are sparsely distributed as in Scenario~3,
the smart LED linking strategy
provides a higher secrecy sum rate than the simple one, 
which can easily be seen by comparing
Fig.~\ref{fig:sec_3d_2_3} and Fig.~\ref{fig:sec_3d_3_3}. 
This is because the eavesdropper
cannot wiretap the messages of all the users,
and the signal interference 
from different LEDs is also mitigated, 
resulting in boosting the secrecy sum rate of VLC systems.

\section{Conclusion}
\label{sec:conclusion}

In this work, 
we investigated the transmission and the secrecy performances of the indoor VLC networks 
for mobile devices equipped with a single PD under NOMA transmission.
Specifically, we proposed a geometric body blockage model with a novel LED arrangement and evaluated our proposed model with  three different transmission strategies, i.e., broadcasting, 
simple LED linking, and smart LED linking strategies. 
As a result,
compared to the conventional broadcasting strategy,
our proposed
\emph{simple} and \emph{smart LED linking strategies}
provided significant improvements in terms of transmission and secrecy performances 
in the simulation results.

For future work, 
the LED arrangement for the different users' locations needs further investigation,
and
it may be of significant interest to derive explicit theoretical bounds of the transmission sum rate and the secrecy sum rate for the body blockage model.

\appendices
\numberwithin{equation}{section} 
\section{Proof of Lemma~\ref{lemma:1}}
\label{sec:appendix}

With reference to \cref{fig:lemma1}, 
a vector on the plane $\mathcal{W}$ and its normal vector ${\bm v}_s$ are orthogonal and thus
\begin{align}
    \label{eq:a1}
    \overrightarrow{NP} \cdot {\bm v}_s = 0.
\end{align} 

The straight line that has the directional vector ${\bm v}_{t}$ and pass through the point $M$ can be expressed as
\begin{align}
    \label{eq:a2}
    \overrightarrow{OP} = \overrightarrow{OM} + \alpha {\bm v}_{t},
\end{align}
where $O$ denotes the origin and $\alpha \in \RR$. 

Substituting \cref{eq:a2} into \cref{eq:a1},
we have
\begin{align*}
    &((\overrightarrow{OM} + \alpha {\bm v}_{t}) - \overrightarrow{ON}) \cdot {\bm v}_s = 0,
\end{align*}
and thus
\begin{align}
    \alpha = \frac{(\overrightarrow{ON}-\overrightarrow{OM}) \cdot {\bm v}_s}{{\bm v}_{t} \cdot {\bm v}_s} = \frac{\overrightarrow{MN} \cdot {\bm v}_s}{{\bm v}_{s} \cdot {\bm v}_t}. 
    \label{eq:a4}
\end{align}
Now substitution of \cref{eq:a4} into \cref{eq:a2} yields that
\begin{align*}
    &\overrightarrow{OP} = \overrightarrow{OM} + \frac{\overrightarrow{MN} \cdot {\bm v}_s}{{\bm v}_s \cdot {\bm v}_{{t}}} {\bm v}_{{t}},
\end{align*}
i.e.,
\begin{align}
    \overrightarrow{MP} = \frac{\overrightarrow{MN} \cdot {\bm v}_s}{{\bm v}_s \cdot {\bm v}_{{t}}} {\bm v}_{{t}},
\end{align}
which completes the proof of Lemma 1. \qed

\bibliographystyle{IEEEtran}
\bibliography{IEEEabrv,shen_citation}

\end{document}